\documentclass[oneside,english,aps]{amsart}
\usepackage[T1]{fontenc}
\usepackage[latin9]{inputenc}
\usepackage{color}
\usepackage{babel}

\usepackage{float}
\usepackage{amsthm}
\usepackage{graphicx}
\usepackage{amssymb}
\usepackage[unicode=true, pdfusetitle,
 bookmarks=true,bookmarksnumbered=true,bookmarksopen=true,bookmarksopenlevel=2,
 breaklinks=false,pdfborder={0 0 1},backref=false,colorlinks=true]
 {hyperref}

\newcommand{\lyxdot}{.}

\numberwithin{equation}{section} 
\numberwithin{figure}{section} 
\theoremstyle{plain}
\theoremstyle{plain}
\newtheorem{thm}{Theorem}
  \theoremstyle{definition}
  \newtheorem{defn}[thm]{Definition}
  \theoremstyle{plain}
  \newtheorem{prop}[thm]{Proposition}
  \theoremstyle{plain}
  \newtheorem{conjecture}[thm]{Conjecture}
  \theoremstyle{plain}
  \newtheorem{cor}[thm]{Corollary}
  \theoremstyle{plain}
  \newtheorem{lem}[thm]{Lemma}
  \theoremstyle{remark}
  \newtheorem{rem}[thm]{Remark}

\begin{document}

\title[Perturbation theory for the Nonlinear Schrödinger Equation...]{Perturbation theory for the Nonlinear Schrödinger Equation with
a random potential}

\author{Shmuel Fishman}

\address{Physics Department, Technion - Israel Institute of Technology, Haifa
32000, Israel.}

\email{fishman@physics.technion.ac.il}

\author{Yevgeny Krivolapov}

\address{Physics Department, Technion - Israel Institute of Technology, Haifa
32000, Israel.}

\email{evgkr@tx.technion.ac.il}

\author{Avy Soffer}

\address{Mathematics Department, Rutgers University, New-Brunswick, NJ 08903,
USA.}

\email{soffer@math.rutgers.edu}

\begin{abstract}
A perturbation theory for the Nonlinear Schrödinger Equation (NLSE)
in 1D on a lattice was developed. The small parameter is the strength
of the nonlinearity. For this purpose secular terms were removed and
a probabilistic bound on small denominators was developed. It was
shown that the number of terms grows exponentially with the order.
The results of the perturbation theory are compared with numerical
calculations. An estimate on the remainder is obtained and it is demonstrated
that the series is asymptotic.
\end{abstract}

\date{January 23rd, 2009}

\maketitle

\section{Introduction}

We consider the problem of dynamical localization of waves in a Nonlinear
Schrödinger Equation (NLSE) \cite{Sulem1999} with a random potential
term on a lattice:\begin{equation}
i\partial_{t}\psi=-J\left[\psi\left(x+1\right)+\psi\left(x-1\right)\right]+\varepsilon_{x}^{\omega}\psi+\beta\left\vert \psi\right\vert ^{2}\psi\label{eq:NLSE}\end{equation}
 where $\psi=\psi\left(x,t\right),$ $x\in\mathbb{Z};$ and $\left\{ \varepsilon_{x}^{\omega}\right\} _{\omega\in\Omega}$
is a collection of i.i.d. random variables chosen from the set $\Omega,$
with the probability measure $\mu\left(\varepsilon_{x}\right).$ It
will be assumed that $\mu\left(\varepsilon\left(x\right)\right)$
is continuous, bounded and of finite support and, additionally, symmetric,
$\mu\left(\varepsilon_{x}\right)=\mu\left(-\varepsilon_{x}\right)$.
We assume that exponential localization is known to take place for
\emph{all} the energies of the linear problem (when $\left.\beta=0\right.$).
The decay rate $\gamma$ and the localization length, $\xi=1/\gamma$,
are given for the linear part of the model \eqref{eq:NLSE} by the
Thouless formula \cite{Thouless1972}. In particular, if $\mu\left(\varepsilon_{x}\right)$
is a uniform distribution than as was found numerically (see Appendix),
the function $\gamma\left(E\right)$ is unimodal.

The NLSE was derived for a variety of physical systems under some
approximations. It was derived in classical optics where $\psi$ is
the electric field by expanding the index of refraction in powers
of the electric field keeping only the leading nonlinear term \cite{Agrawal2007}.
For Bose-Einstein Condensates (BEC), the NLSE is a mean field approximation
where the density $\beta|\psi|^{2}$ approximates the interaction
between the atoms. In this field the NLSE is known as the Gross-Pitaevskii
Equation (GPE) \cite{Dalfovo1999,Pitaevskii2003,Leggett2001,Pitaevskii1961,Gross1961,Pitaevskii1963}.
Recently, it was rigorously established, for a large variety of interactions
and of physical conditions, that the NLSE (or the GPE) is exact in
the thermodynamic limit \cite{Erdos2007,Lieb2002}. Generalized mean
field theories, in which several mean-fields are used, were recently
developed \cite{Cederbaum2003,Alon2005}. In the absence of randomness
(\ref{eq:NLSE}) is completely integrable. For repulsive nonlinearity
$\left(\beta>0\right)$ an initially localized wavepacket spreads,
while for attractive nonlinearity $\left(\beta<0\right)$ solitons
are found typically \cite{Sulem1999}.

It is well known that in 1D in the presence of a random potential
and in the absence of nonlinearity $\left(\beta=0\right)$ with probability
one all the states are exponentially localized \cite{Anderson1958,Ishii1973,Lee1985,Lifshits1988}.
Consequently, diffusion is suppressed and in particular a wavepacket
that is initially localized will not spread to infinity. This is the
phenomenon of Anderson localization. In 2D it is known heuristically
from the scaling theory of localization \cite{Abrahams1979,Lee1985}
that all the states are localized, while in higher dimensions there
is a mobility edge that separates localized and extended states. This
problem is relevant for experiments in nonlinear optics, for example
disordered photonic lattices \cite{Schwartz2007}, where Anderson
localization was found in presence of nonlinear effects as well as
experiments on BECs in disordered optical lattices \cite{Gimperlein2005,Lye2005,Clement2005,Clement2006,Sanchez-Palencia2007,Billy2008,Fort2005,Akkermans2008,Paul2007}.
The interplay between disorder and nonlinear effects leads to new
interesting physics \cite{Fort2005,Akkermans2008,Bishop1995,Rasmussen1999,Kopidakis1999,Kopidakis2000}.
In spite of the extensive research, many fundamental problems are
still open, and, in particular, it is not clear whether in one dimension
(1D) Anderson localization can survive the effects of nonlinearities.
This will be studied here.

A natural question is whether a wave packet that is initially localized
in space will indefinitely spread for dynamics controlled by (\ref{eq:NLSE}).
A simple argument indicates that spreading will be suppressed by randomness.
If unlimited spreading takes place the amplitude of the wave function
will decay since the $L^{2}$ norm is conserved. Consequently, the
nonlinear term will become negligible and Anderson localization will
take place as a result of the randomness. Contrary to this intuition,
based on the smallness of the nonlinear term resulting from the spread
of the wave function, it is claimed that for the kicked-rotor a nonlinear
term leads to delocalization if it is strong enough \cite{Shepelyansky1993}.
It is also argued that the same mechanism results in delocalization
for the model (\ref{eq:NLSE}) with sufficiently large $\beta$, while,
for weak nonlinearity, localization takes place \cite{Shepelyansky1993,Pikovsky2008}.
Therefore, it is predicted in that work that there is a critical value
of $\beta$ that separates the occurrence of localized and extended
states. However, if one applies the arguments of \cite{Shepelyansky1993,Pikovsky2008}
to a variant of \eqref{eq:NLSE}, results that contradict numerical
solutions are found \cite{Mulansky2009,Veksler2009}. Recently, it
was rigorously shown that the initial wavepacket cannot spread so
that its amplitude vanishes at infinite time, at least for large enough
$\beta$ \cite{Kopidakis2008}. It does not contradict spreading of
a fraction of the wavefunction. Indeed, subdiffusion was found in
numerical experiments \cite{Shepelyansky1993,Kopidakis2008,Molina1998}.
In different works \cite{Molina1998,Flach2009,Skokos2009} sub-diffusion
was reported for all values of $\beta$, but with a different power
of the time dependence (compared with Ref.~\cite{Shepelyansky1993}).
It was also argued that nonlinearity may enhance discrete breathers
\cite{Kopidakis1999,Kopidakis2000}. In conclusion, it is \emph{not}
clear what is the long time behavior of a wave packet that is initially
localized, if both nonlinearity and disorder are present. This is
the main motivation for the present work. Since heuristic arguments
and numerical simulations produce conflicting results, rigorous statements
are required for further progress.

More precisely, the question of dynamical localization can be rigorously
formulated as follows: assume the initial state is, $\psi\left(x,0\right)\equiv u_{0}\left(x\right),$
where $u_{0}\left(x\right)$ is an eigenstate of the linear part of
\eqref{eq:NLSE} which is localized near $x=0$. Then for any $0<\delta<1,$
one has to prove that with probability $1-\delta$ (on the space of
the potentials)\begin{equation}
\sup_{x,t}\left\vert e^{\nu\left\vert x\right\vert }\psi\left(x,t\right)\right\vert <M_{\delta}<\infty\end{equation}
 for some $\nu>0.$

Rigorous results on dynamical localization for the linear case are
well known \cite{Combes1994,Germinet1998,Klopp1995}. However, the
nonlinear problem turns out to be very difficult to handle, even numerically.
Consider the case of small $\beta$. There are two possible mechanisms
for destruction of the localization due to nonlinearity.

One way of spreading is to spread into many random places with increasing
number of them. In that case, due to conservation of the normalization
of the solution, the solution becomes small. But then, the nonlinear
term becomes less and less important and we expect the linear theory
to take over and lead to localization. While this argument sounds
plausible there is no proof along this lines.

The second way of spreading is in a few fixed number of spikes that
hop randomly to infinity. In this case, the nonlinear term is always
relevant. It is this (possible) process that makes the proof of localization
in the nonlinear case so elusive. It also precludes a quick numerical
analysis of the problem: it may take exponentially long time to see
the hoping.

Rigorous results in this direction are of preliminary nature: In \cite{Soffer2003}
it was shown that dynamical localization holds for the linear problem
perturbed by a periodic in time and exponentially localized in space
small linear perturbation. In \cite{Bourgain2004} the above result
was extended to a quasiperiodic in time perturbation. Such perturbations
mimic the nonlinear term:\begin{equation}
\left\vert \psi\right\vert ^{2}\rightarrow\left\vert {\displaystyle \sum\limits _{j}}c_{j}u_{j}\left(x\right)e^{iE_{j}t}\right\vert ^{2}\end{equation}
 where $u_{j}$ are the eigenfunctions of the linear problem with
energies $E_{j}.$ However in other situations time dependent terms
may result in delocalization \cite{Wang2008a,Nersesyan2008}. Using
normal form transformations Wang and Zhang \cite{Wang2008} studied
the limit of strong disorder and weak nonlinearity, namely, $\epsilon=J+\beta$,
small. For initial wavefunctions with tails of weight $\delta$ starting
from point $j_{0}$, they have proved, that the wavefunction spreads
as following. There exist $C=C\left(A\right)>0$ and $\epsilon\left(A\right)>0$
and $K=K\left(A\right)>A^{2}$ such that for all $t\leq\left(\delta/C\right)\epsilon^{-A}$
the weight of the tails of the spreaded wavefunction starting from
$j=j_{0}+K$ is less that $2\delta$. On the basis of this result
they have conjectured that the spread of the wave function is at most
logarithmic in $t$.

Furthermore, it can be shown that NLSE has stationary solutions\begin{equation}
E\psi_{E}=\left(-\partial_{xx}+\varepsilon_{x}^{\omega}\right)\psi_{E}+\beta\left\vert \psi_{E}\right\vert ^{2}\psi_{E}\end{equation}
 which are exponentially localized for almost all $E$ with a localization
length that is identical to the one of the linear problem \cite{Albanese1988,Albanese1988a,Albanese1991,Iomin2007,Fishman2008}.

In our previous work \cite{Fishman2008a} we have developed a perturbation
theory in $\beta$. By considering the first order expansion we have
proved that for times of order $O\left(\beta^{-2}\right)$ the solution
of \eqref{eq:NLSE} remains exponentially localized. A result of similar
nature for a nonlinear equation of a different structure was obtained
in \cite{Benettin1988}. In the current work we consider an expansion
of any order, $N$, in $\beta$. This expansion enables\emph{ in principle
}the calculation of the solution to any order in $N$. A bound on
the error can be computed using only propreties of the linear problem
$\left(\beta=0\right)$. Therefore this work has the potential to
develop into a method for solution of some type of nonlinear differential
equations. In Section \ref{sec:Organization-of-the} we construct
the solution as a series in the eigenfunctions of the linear problem.
Standard perturbation theory for the coefficients does not apply:
we encounter small divisor problems and secular terms (formally infinite).
Removing the secular terms requires the {}``renormalization'' of
the original linear Hamiltonian by shifting the energies (Section
\ref{sec:Elimination-of-secular}). The estimates of the small divisor
terms are performed in the spirit of the work of Aizenman-Molchanov
(A-M) \cite{Aizenman1993}. In Section \ref{sec:The-entropy-problem}
the entropy problem is resolved by bounding an appropriate recursive
relation. A general probabilistic bound on the terms of the perturbation
theory is derived in Section \ref{sec:Bounding-the-general} and the
quality of the perturbation theory is tested in Section \ref{sec:Numerical-results}.
In Section \ref{sec:Bounding-the-remainder} the remainder terms are
controlled by a bootstrap argument. The results are summarized in
Section \ref{sec:Summary} and the open problems are listed there.

In summary, in this work a perturbation theory for \eqref{eq:NLSE}
in powers of $\beta$ was developed and bounds on the various terms
were obtained. \emph{The work is only partly rigorous.} In some parts
it relies on Conjectures that we test numerically.

\section{\label{sec:Organization-of-the}Organization of the perturbation
theory}

Our goal is to analyze the nonlinear Schrödinger equation\begin{equation}
i\partial_{t}\psi=H_{0}\psi+\beta\left\vert \psi\right\vert ^{2}\psi\label{eq:GPE}\end{equation}
 where $H_{0}$ is the Anderson Hamiltonian,\begin{equation}
H_{0}\psi\left(x\right)=-J\left[\psi\left(x+1\right)+\psi\left(x-1\right)\right]+\varepsilon_{x}\psi\left(x\right).\end{equation}
We assume throughout the paper that $H_{0}$ satisfies the conditions
for localization, namely, for almost all the realizations, $\omega,$
of the disordered potential, all the eigenstates of $H_{0},$ $u_{m},$
are exponentially localized and have an envelope of the form of\begin{equation}
\left\vert u_{m}\left(x\right)\right\vert \leq D_{\omega,\varepsilon}e^{\varepsilon\left\vert x_{m}\right\vert }e^{-\gamma\left\vert x-x_{m}\right\vert },\label{eq:loc_states}\end{equation}
 where $\varepsilon>0,$ $x_{m}$ is the localization center which
will be defined at the next subsection, $\gamma$ is the inverse of
the localization length, $\xi=\gamma^{-1},$ and $D_{\omega,\varepsilon}$
is a constant dependent on $\varepsilon$ and the realization of the
disordered potential \cite{Rio1995,Rio1996} (better estimates were
proven recently in \cite{DeBi`evre2000,Germinet2006}). It is of importance
that $D_{\omega,\varepsilon}$ does not depend on the energy of the
state. In the present work only realizations $\omega,$ where\begin{equation}
\left\vert D_{\omega,\varepsilon}\right\vert \leq D_{\delta,\varepsilon}<\infty\label{eq:real_assump}\end{equation}
 are considered. This is satisfied for a set of a measure $\left.1-\delta\right.$,
since \eqref{eq:loc_states} is false only for a measure zero of potentials.

\subsection{Assignment of eigenfunctions to sites}

It is tempting to assign eigenfunctions to sites by their maxima,
namely, $u_{E_{i}}$ is the eigenfunction with energy $E_{i}$ and
a maximum at site $i$. This assignment is very unstable with respect
to the change of realizations. This is due to the fact that the point
where the maximum is found, which is sometimes called the localization
center, may change as a result of a very small change in the on-site
energies $\left\{ \varepsilon_{x}\right\} $. To avoid this, the assignment
is defined as the center of mass \cite{Fishman2008b},\begin{equation}
x_{E}=\sum\nolimits _{x}x\left|u_{E_{i}}\left(x\right)\right|^{2}\end{equation}

\begin{defn}
\label{def:assignment_to_sites}The state $u_{E_{i}}$ is assigned
to site $i$ if $i=\left[x_{E}\right].$ If several states are assigned
to the same site we order them by energy.
\end{defn}

\subsection{The perturbation expansion}

The wavefunction can be expanded using the eigenstates of $H_{0}$
as \begin{equation}
\psi\left(x,t\right)=\sum_{m}c_{m}\left(t\right)e^{-iE_{m}t}u_{m}\left(x\right).\label{eq:expansion}\end{equation}
For the nonlinear equation the dependence of the expansion coefficients,
$c_{m}\left(t\right),$ is found by inserting this expansion into
(\ref{eq:GPE}), resulting in\begin{align}
i\partial_{t}\sum_{m}c_{m}e^{-iE_{m}t}u_{m}\left(x\right) & =H_{0}\sum_{m}c_{m}e^{-iE_{m}t}u_{m}\left(x\right)\\
 & +\beta\left\vert \sum_{m}c_{m}e^{-iE_{m}t}u_{m}\left(x\right)\right\vert ^{2}\sum_{m_{3}}c_{m_{3}}e^{-iE_{m_{3}}t}u_{m_{3}}\left(x\right).\nonumber \end{align}
 Multiplying by $u_{n}\left(x\right)$ and integrating gives\begin{equation}
i\partial_{t}c_{n}=\beta\sum_{m_{1},m_{2},m_{3}}V_{n}^{m_{1}m_{2}m_{3}}c_{m_{1}}^{\ast}c_{m_{2}}c_{m_{3}}e^{i\left(E_{m_{1}}+E_{n}-E_{m_{2}}-E_{m_{3}}\right)t}\label{eq:c_n_exact}\end{equation}
 where $V_{n}^{m_{1}m_{2}m_{3}}$ is an overlap sum\begin{equation}
V_{n}^{m_{1}m_{2}m_{3}}=\sum_{x}u_{n}\left(x\right)u_{m_{1}}\left(x\right)u_{m_{2}}\left(x\right)u_{m_{3}}\left(x\right).\label{eq:overlap_int}\end{equation}
By definition $V_{n}^{m_{1}m_{2}m_{3}}$ is symmetric with respect
to an interchange of any two indices. Additionally, since the $u_{n}\left(x\right)$
are exponentially localized around $x_{n},$ $V_{n}^{m_{1}m_{2}m_{3}}$
is not negligible only when the interval, \begin{equation}
\delta m\equiv\max\left[x_{n},x_{m_{i}}\right]-\min\left[x_{n},x_{m_{i}}\right],\end{equation}
 is of the order of the localization length, around $x_{n},$\begin{align}
\left\vert V_{n}^{m_{1}m_{2}m_{3}}\right\vert  & \leq D_{\delta,\varepsilon}^{4}e^{\varepsilon\left(\left\vert x_{n}\right\vert +\left\vert x_{m_{1}}\right\vert +\left\vert x_{m_{2}}\right\vert +\left\vert x_{m_{3}}\right\vert \right)}\sum_{x}\cdot e^{-\gamma\left(\left\vert x-x_{n}\right\vert +\left\vert x-x_{m_{1}}\right\vert +\left\vert x-x_{m_{2}}\right\vert +\left\vert x-x_{m_{3}}\right\vert \right)}\label{eq:overlap}\\
 & \leq D_{\delta,\varepsilon}^{4}e^{\varepsilon\left(\left\vert x_{n}\right\vert +\left\vert x_{m_{1}}\right\vert +\left\vert x_{m_{2}}\right\vert +\left\vert x_{m_{3}}\right\vert \right)}e^{-\frac{\left(\gamma-\varepsilon^{\prime}\right)}{3}\left(\left\vert x_{n}-x_{m_{1}}\right\vert +\left\vert x_{n}-x_{m_{2}}\right\vert +\left\vert x_{n}-x_{m_{3}}\right\vert \right)}\times\nonumber \\
 & \times\sum_{x}\cdot e^{-\varepsilon^{\prime}\left(\left\vert x-x_{n}\right\vert +\left\vert x-x_{m_{1}}\right\vert +\left\vert x-x_{m_{2}}\right\vert +\left\vert x-x_{m_{3}}\right\vert \right)}\nonumber \\
 & \leq V_{\delta}^{\varepsilon,\varepsilon^{\prime}}e^{\varepsilon\left(\left\vert x_{n}\right\vert +\left\vert x_{m_{1}}\right\vert +\left\vert x_{m_{2}}\right\vert +\left\vert x_{m_{3}}\right\vert \right)}e^{-\frac{1}{3}\left(\gamma-\varepsilon^{\prime}\right)\left(\left\vert x_{n}-x_{m_{1}}\right\vert +\left\vert x_{n}-x_{m_{2}}\right\vert +\left\vert x_{n}-x_{m_{3}}\right\vert \right)}.\nonumber \end{align}
Here we have used the triangle inequality\begin{align}
\left(\left\vert x-x_{n}\right\vert +\left\vert x-x_{m_{1}}\right\vert \right)+ & \left(\left\vert x-x_{n}\right\vert +\left\vert x-x_{m_{2}}\right\vert \right)+\label{eq:first_triangle}\\
+\left(\left\vert x-x_{n}\right\vert +\left\vert x-x_{m_{3}}\right\vert \right) & \geq\left\vert x_{n}-x_{m_{1}}\right\vert +\left\vert x_{n}-x_{m_{2}}\right\vert +\left\vert x_{n}-x_{m_{3}}\right\vert \nonumber \end{align}
 to obtain the second line. Our objective is to develop a perturbation
expansion of the $c_{m}\left(t\right)$ in powers of $\beta$ and
to calculate them order by order in $\beta.$ The required expansion
is\begin{equation}
c_{n}\left(t\right)=c_{n}^{\left(0\right)}+\beta c_{n}^{\left(1\right)}+\beta^{2}c_{n}^{\left(2\right)}+\cdots+\beta^{N-1}c_{n}^{\left(N-1\right)}+\beta^{N}Q_{n},\label{eq:cn_expand}\end{equation}
 where the expansion is till order $\left(N-1\right)$ and $Q_{n}$
is the remainder term. We will assume the initial condition\begin{equation}
c_{n}\left(t=0\right)=\delta_{n0}.\end{equation}
 The equations for the two leading orders are presented in what follows.
The leading order is\begin{equation}
c_{n}^{\left(0\right)}=\delta_{n0}.\label{eq:cn0}\end{equation}
The equation for the first order is\begin{equation}
i\partial_{t}c_{n}^{\left(1\right)}=\sum_{m_{1},m_{2},m_{3}}V_{n}^{m_{1}m_{2}m_{3}}c_{m_{1}}^{\ast\left(0\right)}c_{m_{2}}^{\left(0\right)}c_{m_{3}}^{\left(0\right)}e^{i\left(E_{n}+E_{m_{1}}-E_{m_{2}}-E_{m_{3}}\right)t}=V_{n}^{000}e^{i\left(E_{n}-E_{0}\right)t}\end{equation}
 and its solution is\begin{equation}
c_{n}^{\left(1\right)}=V_{n}^{000}\left(\frac{1-e^{i\left(E_{n}-E_{0}\right)t}}{E_{n}-E_{0}}\right).\end{equation}
 The resulting equation for the second order is\begin{align}
i\partial_{t}c_{n}^{\left(2\right)} & =\sum_{m_{1},m_{2},m_{3}}V_{n}^{m_{1}m_{2}m_{3}}c_{m_{1}}^{\ast\left(1\right)}c_{m_{2}}^{\left(0\right)}c_{m_{3}}^{\left(0\right)}e^{i\left(E_{n}+E_{m_{1}}-E_{m_{2}}-E_{m_{3}}\right)t}+\\
 & +2\sum_{m_{1},m_{2},m_{3}}V_{n}^{m_{1}m_{2}m_{3}}c_{m_{1}}^{\ast\left(0\right)}c_{m_{2}}^{\left(1\right)}c_{m_{3}}^{\left(0\right)}e^{i\left(E_{n}+E_{m_{1}}-E_{m_{2}}-E_{m_{3}}\right)t}.\nonumber \end{align}
 Substitution of the lower orders yields\begin{align}
i\partial_{t}c_{n}^{\left(2\right)} & =\sum_{m}V_{n}^{m00}V_{m}^{000}\left[\left(\frac{1-e^{-i\left(E_{m}-E_{0}\right)t}}{E_{m}-E_{0}}\right)e^{i\left(E_{n}+E_{m}-2E_{0}\right)t}+\right.\\
 & \left.+2\left(\frac{1-e^{i\left(E_{m}-E_{0}\right)t}}{E_{m}-E_{0}}\right)e^{i\left(E_{n}-E_{m}\right)t}\right]\nonumber \\
 & =\sum_{m}\frac{V_{n}^{m00}V_{m}^{000}}{E_{m}-E_{0}}\left[\left(e^{i\left(E_{n}+E_{m}-2E_{0}\right)t}-e^{i\left(E_{n}-E_{0}\right)t}\right)+\right.\nonumber \\
 & \left.+2\left(e^{i\left(E_{n}-E_{m}\right)t}-e^{i\left(E_{n}-E_{0}\right)t}\right)\right]\nonumber \\
 & =\sum_{m}\frac{V_{n}^{m00}V_{m}^{000}}{E_{m}-E_{0}}\left[e^{i\left(E_{n}+E_{m}-2E_{0}\right)t}-3e^{i\left(E_{n}-E_{0}\right)t}+2e^{i\left(E_{n}-E_{m}\right)t}\right].\nonumber \end{align}
 We notice that divergence of this expansion for any value of $\beta$
may result from three major problems: the secular terms problem, the
entropy problem (i.e., factorial proliferation of terms), and the
small denominators problem.

\section{\label{sec:Elimination-of-secular}Elimination of secular terms}

We first show how to derive the equations for $c_{n}\left(t\right)$
where the secular terms are eliminated.
\begin{prop}
To each order in $\beta,$ $\psi\left(x,t\right)$ can be expanded
as \begin{equation}
\psi\left(x,t\right)={\displaystyle \sum\limits _{n}}c_{n}\left(t\right)e^{-iE_{n}^{\prime}t}u_{n}\left(x\right)\end{equation}
 with\begin{equation}
E_{n}^{\prime}\equiv E_{n}^{\left(0\right)}+\beta E_{n}^{\left(1\right)}+\beta^{2}E_{n}^{\left(2\right)}+\cdots\label{eq:En_expand}\end{equation}
 and $E_{n}^{\left(0\right)}$ are the eigenvalues of $H_{0},$ in
such a way that there are no secular terms to any given order. The
$E_{n}'$ are called the renormalized energies.
\end{prop}
Here we first develop the general scheme for the elimination of the
secular terms and then demonstrate the construction of $E_{n}^{\prime}$
when the $c_{n}\left(t\right)$ are calculated to the second order
in $\beta$ (see~\ref{eq:cn_order2},\ref{eq:en_oder2}).

Inserting the expansion into (\ref{eq:GPE}) yields\begin{align}
i{\displaystyle \sum\limits _{m}}\left[\partial_{t}c_{m}-iE_{m}^{\prime}c_{m}\right]e^{-iE_{m}^{\prime}t}u_{m}\left(x\right) & ={\displaystyle \sum\limits _{m}}E_{m}^{\left(0\right)}c_{m}e^{-iE_{m}^{\prime}t}u_{m}\left(x\right)+\\
+\beta{\displaystyle \sum\limits _{m_{1}m_{2}m_{3}}}c_{m_{1}}^{\ast}c_{m_{2}}c_{m_{3}} & e^{i\left(E_{m_{1}}^{\prime}-E_{m_{2}}^{\prime}-E_{m_{3}}^{\prime}\right)t}u_{m_{1}}\left(x\right)u_{m_{2}}\left(x\right)u_{m_{3}}\left(x\right).\nonumber \end{align}
 Multiplication by $u_{n}\left(x\right)$ and integration gives\begin{equation}
i\partial_{t}c_{n}=\left(E_{n}^{\left(0\right)}-E_{n}^{\prime}\right)c_{n}+\beta\sum_{m_{1}m_{2}m_{3}}V_{n}^{m_{1}m_{2}m_{3}}c_{m_{1}}^{\ast}c_{m_{2}}c_{m_{3}}e^{i\left(E_{n}^{\prime}+E_{m_{1}}^{\prime}-E_{m_{2}}^{\prime}-E_{m_{3}}^{\prime}\right)t},\label{eq:diff_eq}\end{equation}
 where the $V_{n}^{m_{1}m_{2}m_{3}}$ are given by (\ref{eq:overlap_int}).
Following (\ref{eq:cn_expand}) we expand $c_{n}$ in orders of $\beta,$
namely,\begin{equation}
c_{n}=c_{n}^{\left(0\right)}+\beta c_{n}^{\left(1\right)}+\beta^{2}c_{n}^{\left(2\right)}+\cdots,.\end{equation}
 Inserting this expansion into \eqref{eq:diff_eq} and comparing the
powers of $\beta$ without expanding the exponent, produces the following
equation for the $k-th$ order\begin{align}
i\partial_{t}c_{n}^{\left(k\right)} & =-\sum_{s=0}^{k-1}E_{n}^{\left(k-s\right)}c_{n}^{\left(s\right)}+\label{eq:a2}\\
 & +\sum_{m_{1}m_{2}m_{3}}V_{n}^{m_{1}m_{2}m_{3}}\left[\sum_{r=0}^{k-1}\sum_{s=0}^{k-1-r}\sum_{l=0}^{k-1-r-s}c_{m_{1}}^{\left(r\right)\ast}c_{m_{2}}^{\left(s\right)}c_{m_{3}}^{\left(l\right)}\right]e^{i\left(E_{n}^{\prime}+E_{m_{1}}^{\prime}-E_{m_{2}}^{\prime}-E_{m_{3}}^{\prime}\right)t}.\nonumber \end{align}
Note that the exponent is of order $O\left(1\right)$ in $\beta$,
and therefore we may choose not to expand it in powers of $\beta$.
However, it generates an expansion where both $E{}_{m}^{\left(l\right)}$
and $c_{n}^{\left(k\right)}$ depend on $\beta$. For the expansion
\eqref{eq:cn_expand} to be valid, both $E{}_{m}^{\left(l\right)}$
and $c_{n}^{\left(k\right)}$ should be $O\left(1\right)$ in $\beta$,
this is satisfied since the RHS of \eqref{eq:a2} contains only $c_{n}^{\left(r\right)}$
such that $r<k$. Namely, this equation gives each order in terms
of the lower ones, with the initial condition of $\left.c_{n}^{\left(0\right)}\left(t\right)=\delta_{n0}\right..$
Solution of $k$ equations \eqref{eq:a2} gives the solution of the
differential equation \eqref{eq:diff_eq} to order $k$. Since, the
exponent in \eqref{eq:a2} is of order $O\left(1\right)$ in $\beta$
we can select its argument to be of any order in $\beta$. However,
for the removal of the secular terms, as will be explained bellow,
it is instructive to set the order of the argument to be $k-1$, as
the higher orders were not calculated at this stage. Secular terms
are created when there are time independent terms in the RHS of the
equation above. We eliminate those terms by using the first two terms
in the first summation on the RHS. We make use of the fact that $c_{n}^{\left(0\right)}=\delta_{n0}$
and $c_{n}^{\left(1\right)}$ can be easily determined (see (\ref{eq:c01},\ref{eq:cn1})),
and used to calculate $E_{n=0}^{\left(k\right)}$ and $E_{n\neq0}^{\left(k-1\right)}$
that eliminate the secular terms in the equation for $c_{n}^{\left(k\right)},$
that is \begin{equation}
E_{n}^{\left(k\right)}c_{n}^{\left(0\right)}+E_{n}^{\left(k-1\right)}c_{n}^{\left(1\right)}=E_{n}^{\left(k\right)}\delta_{n0}+E_{n}^{\left(k-1\right)}\left(1-\delta_{n0}\right)\frac{V_{n}^{000}}{E_{n}^{\prime}-E_{0}^{\prime}},\label{eq:sec_elim}\end{equation}
 where only the time-independent part of $c_{n}^{\left(1\right)}$
was used. In other words, we choose $E_{n}^{\left(k\right)}$ and
$E_{n\neq0}^{\left(k-1\right)}$ so that the time-independent terms
on the RHS of (\ref{eq:a2}) are eliminated. $E_{0}^{\left(k\right)}$
will eliminate all secular terms with $n=0,$ and $E_{n}^{\left(k-1\right)}$
will eliminate all secular terms with $n\neq0.$ In the following,
we will demonstrate this procedure for the first two orders, and calculate
$c_{n}^{\left(1\right)}$, and obtain an equation for $c_{n}^{\left(2\right)}$.

In the first order of the expansion in $\beta$ we obtain\begin{align}
i\partial_{t}c_{n}^{\left(1\right)} & =-E_{n}^{\left(1\right)}c_{n}^{\left(0\right)}+\sum_{m_{1}m_{2}m_{3}}V_{n}^{m_{1}m_{2}m_{3}}c_{m_{1}}^{\ast\left(0\right)}c_{m_{2}}^{\left(0\right)}c_{m_{3}}^{\left(0\right)}e^{i\left(E_{n}'+E_{m_{1}}'-E_{m_{2}}'-E_{m_{3}}'\right)t}\\
 & =-E_{n}^{\left(1\right)}\delta_{n0}+V_{n}^{000}e^{i\left(E_{n}'-E_{0}'\right)t}.\nonumber \end{align}
 For $n=0$ the equation produces a secular term\begin{align}
i\partial_{t}c_{0}^{\left(1\right)} & =-E_{0}^{\left(1\right)}+V_{0}^{000}\label{eq:c01}\\
c_{0}^{\left(1\right)} & =it\cdot\left(E_{0}^{\left(1\right)}-V_{0}^{000}\right).\nonumber \end{align}
 Setting \begin{equation}
E_{0}^{\left(1\right)}=V_{0}^{000}\end{equation}
 will eliminate this secular term and gives \begin{equation}
c_{0}^{\left(1\right)}=0\end{equation}
 For $n\neq0$ there are no secular terms in this order, therefore
finally\begin{equation}
c_{n}^{\left(1\right)}=\left(1-\delta_{n0}\right)V_{n}^{000}\left(\frac{1-e^{i\left(E_{n}'-E_{0}'\right)t}}{E_{n}'-E_{0}'}\right),\label{eq:cn1}\end{equation}
where to this order $E_{n}'=E_{n}$ and $E'_{0}=E_{0}$.

In the second order of the expansion in $\beta$ we have\begin{align}
i\partial_{t}c_{n}^{\left(2\right)} & =-E_{n}^{\left(1\right)}c_{n}^{\left(1\right)}-E_{n}^{\left(2\right)}c_{n}^{\left(0\right)}+\\
 & +\sum_{m_{1}m_{2}m_{3}}V_{n}^{m_{1}m_{2}m_{3}}\left(c_{m_{1}}^{\ast\left(1\right)}c_{m_{2}}^{\left(0\right)}c_{m_{3}}^{\left(0\right)}+2c_{m_{1}}^{\ast\left(0\right)}c_{m_{2}}^{\left(1\right)}c_{m_{3}}^{\left(0\right)}\right)e^{i\left(E_{n}^{\prime}+E_{m_{1}}^{\prime}-E_{m_{2}}^{\prime}-E_{m_{3}}^{\prime}\right)t}\nonumber \\
 & =-E_{n}^{\left(2\right)}\delta_{n0}-E_{n}^{\left(1\right)}c_{n}^{\left(1\right)}+\sum_{m_{1}}V_{n}^{m_{1}00}\left(c_{m_{1}}^{\ast\left(1\right)}e^{i\left(E_{n}^{\prime}+E_{m_{1}}^{\prime}-2E_{0}^{\prime}\right)t}+2c_{m_{1}}^{\left(1\right)}e^{i\left(E_{n}^{\prime}-E_{m_{1}}^{\prime}\right)t}\right).\nonumber \end{align}
For $n=0$ it takes the form\[
i\partial_{t}c_{0}^{\left(2\right)}=-E_{0}^{\left(2\right)}+\sum_{m}V_{0}^{m00}\left(c_{m}^{\ast\left(1\right)}e^{i\left(E_{m}^{\prime}-E_{0}^{\prime}\right)t}+2c_{m}^{\left(1\right)}e^{i\left(E_{0}^{\prime}-E_{m}^{\prime}\right)t}\right).\]
 Substitution of (\ref{eq:c01}) and \eqref{eq:cn1} yields\begin{align}
i\partial_{t}c_{0}^{\left(2\right)} & =-E_{0}^{\left(2\right)}+\sum_{m\neq0}\frac{V_{0}^{m00}V_{m}^{000}}{E_{m}^{\prime}-E_{0}^{\prime}}\left[\left(1-e^{-i\left(E_{m}^{\prime}-E_{0}^{\prime}\right)t}\right)e^{i\left(E_{m}^{\prime}-E_{0}^{\prime}\right)t}+\right.\\
 & \left.+2\left(1-e^{i\left(E_{m}^{\prime}-E_{0}^{\prime}\right)t}\right)e^{i\left(E_{0}^{\prime}-E_{m}^{\prime}\right)t}\right]\nonumber \\
 & =-E_{0}^{\left(2\right)}+\sum_{m\neq0}\frac{V_{0}^{m00}V_{m}^{000}}{E_{m}^{\prime}-E_{0}^{\prime}}\left(e^{i\left(E_{m}^{\prime}-E_{0}^{\prime}\right)t}+2e^{i\left(E_{0}^{\prime}-E_{m}^{\prime}\right)t}-3\right),\nonumber \end{align}
 and the secular term could be removed by setting\begin{equation}
E_{0}^{\left(2\right)}=-3\sum_{m\neq0}\frac{V_{0}^{m00}V_{m}^{000}}{E_{m}^{\prime}-E_{0}^{\prime}}.\end{equation}
 For $n\neq0$ we have\begin{align}
i\partial_{t}c_{n}^{\left(2\right)} & =-E_{n}^{\left(1\right)}V_{n}^{000}\left(\frac{1-e^{i\left(E_{n}^{\prime}-E_{0}^{\prime}\right)t}}{E_{n}^{\prime}-E_{0}^{\prime}}\right)+\\
 & +\sum_{m}V_{n}^{m00}\left(c_{m}^{\ast\left(1\right)}e^{i\left(E_{n}^{\prime}+E_{m}^{\prime}-2E_{0}^{\prime}\right)t}+2c_{m}^{\left(1\right)}e^{i\left(E_{n}^{\prime}-E_{m}^{\prime}\right)t}\right)\nonumber \\
= & -E_{n}^{\left(1\right)}V_{n}^{000}\left(\frac{1-e^{i\left(E_{n}^{\prime}-E_{0}^{\prime}\right)t}}{E_{n}^{\prime}-E_{0}^{\prime}}\right)+\nonumber \\
 & +\sum_{m\neq0}\frac{V_{n}^{m00}V_{m}^{000}}{E_{m}^{\prime}-E_{0}^{\prime}}\left(e^{i\left(E_{n}^{\prime}+E_{m}^{\prime}-2E_{0}^{\prime}\right)t}+2e^{i\left(E_{n}^{\prime}-E_{m}^{\prime}\right)t}-3e^{i\left(E_{n}^{\prime}-E_{0}^{\prime}\right)t}\right).\nonumber \end{align}
 We notice that the second term in the sum produces secular terms
for $m=n.$ Those terms could be removed by setting\begin{equation}
-\frac{E_{n}^{\left(1\right)}V_{n}^{000}}{E_{n}^{\prime}-E_{0}^{\prime}}+\frac{2V_{n}^{n00}V_{n}^{000}}{E_{n}^{\prime}-E_{0}^{\prime}}=0\qquad n\neq0\end{equation}
\[
E_{n}^{\left(1\right)}=2V_{n}^{n00}\qquad n\neq0.\]
 To conclude, up to the second order in $\beta$ , the perturbed energies,
which are required to remove the secular terms, are given by\begin{equation}
E_{n}^{\prime}=E_{n}^{\left(0\right)}+\beta V_{n}^{n00}\left(2-\delta_{n0}\right)-3\beta^{2}\delta_{n0}\sum_{m\neq0}\frac{\left(V_{m}^{000}\right)^{2}}{E_{m}^{\prime}-E_{0}^{\prime}},\label{eq:en_oder2}\end{equation}
 and the corresponding correction to $c_{n}^{\left(0\right)}$ is\begin{equation}
i\partial_{t}c_{n}^{\left(2\right)}=\left\{ \begin{array}{cc}
\sum_{m\neq0}\frac{V_{0}^{m00}V_{m}^{000}}{E_{m}^{\prime}-E_{0}^{\prime}}\left(e^{i\left(E_{m}^{\prime}-E_{0}^{\prime}\right)t}+2e^{i\left(E_{0}^{\prime}-E_{m}^{\prime}\right)t}\right) & n=0\\
\begin{array}{c}
\frac{2V_{n}^{n00}V_{n}^{000}}{E_{n}^{\prime}-E_{0}^{\prime}}e^{i\left(E_{n}^{\prime}-E_{0}^{\prime}\right)t}+\sum_{m\neq0,n}\frac{2V_{n}^{m00}V_{m}^{000}}{E_{m}^{\prime}-E_{0}^{\prime}}e^{i\left(E_{n}^{\prime}-E_{m}^{\prime}\right)t}+\\
+\sum_{m\neq0}\frac{V_{n}^{m00}V_{m}^{000}}{E_{m}^{\prime}-E_{0}^{\prime}}\left(e^{i\left(E_{n}^{\prime}+E_{m}^{\prime}-2E_{0}^{\prime}\right)t}-3e^{i\left(E_{n}^{\prime}-E_{0}^{\prime}\right)t}\right)\end{array} & n\neq0\end{array}\right..\label{eq:cn_order2}\end{equation}
 Note that in the calculation of $c_{n}$ to higher orders in $\beta,$
a secular term of the order $\beta^{2}$ will be generated for $\left.n\neq0\right..$
Secular terms with increasing complexity are generated in the cancellation
of higher orders, however, as demonstrated by (\ref{eq:sec_elim}),
secular terms are removed with the same $c_{n}^{\left(0\right)}$
and $c_{n}^{\left(1\right)}$ which are presented in (\ref{eq:cn0},\ref{eq:cn1}).

In the next section, the entropy problem will be studied. It will
be shown that the proliferation of terms in the expansion is at most
exponential.

\section{\label{sec:The-entropy-problem}The entropy problem}

Since the time dependence of all orders is bounded (excluding the
secular terms), we can bound each order of the expansion by\begin{align}
\left\vert c_{n}^{\left(0\right)}\right\vert  & =\delta_{n0}\label{eq:first_bounds}\\
\left\vert c_{n}^{\left(1\right)}\right\vert  & =\left\vert V_{n}^{000}\left(\frac{1-e^{i\left(E_{n}'-E_{0}'\right)t}}{E_{n}'-E_{0}'}\right)\right\vert \leq2\left\vert \frac{V_{n}^{000}}{E_{n}'-E_{0}'}\right\vert \nonumber \\
\left\vert c_{n}^{\left(2\right)}\right\vert  & =\left\vert \sum_{m}\int_{0}^{t}dt'\left(\frac{V_{n}^{m00}V_{m}^{000}}{E_{m}'-E_{0}'}\right)\left(e^{i\left(E_{n}'+E_{m}'-2E_{0}'\right)t'}-3e^{i\left(E_{n}'-E_{0}'\right)t'}+2e^{i\left(E_{n}'-E_{m}'\right)t'}\right)\right\vert \nonumber \\
 & \leq2\cdot\sum\limits _{m}\frac{\left\vert V_{n}^{m00}\right\vert \left\vert V_{m}^{000}\right\vert }{\left\vert E_{m}'-E_{0}'\right\vert }\left(\frac{1}{\left\vert \left(E_{n}'+E_{m}'-2E_{0}'\right)\right\vert }+\frac{3}{\left\vert E_{n}'-E_{0}'\right\vert }+\frac{2}{\left\vert E_{n}'-E_{m}'\right\vert }\right)\nonumber \end{align}
 et cetera. However, for convergence for a finite but possibly small
$\beta$, it is essential that the number of terms on the RHS of (\ref{eq:first_bounds})
will not increase faster than exponentially in $k,$ e.g. not as $k!,$
where $k$ is the expansion order. Next we will show that the number
of terms indeed increases at most exponentially in $k.$

We will designate the number of different products of order $k$ of
$V^{\prime}s$ by $R_{k}$ (on top of it there is still a number of
non vanishing terms in the sums over $m,$ that will be estimated
in the next section). By replacing each $c_{n}^{\left(l\right)}$
in \eqref{eq:a2} by $R_{l}$ (the integration with respect to time
multiples the number of terms by a factor of 2, cf. \eqref{eq:first_bounds})
we deduce a recursive expression for $R_{k}$\begin{equation}
R_{k}=\sum_{r=0}^{k-1}\sum_{l=0}^{k-1-r}R_{r}R_{l}R_{k-1-r-l}\qquad R_{1}=1\qquad R_{0}=1.\label{eq:recurrence}\end{equation}

In order to find an upper bound on $R_{k}$ we examine the structure
of the products of $V^{\prime}s$ we notice that each product could
be uniquely labeled by a vector of zeros and $m_{i}^{\prime}s$\begin{equation}
V_{n}^{m_{1}00}V_{m_{1}}^{m_{2}m_{3}0}\cdots V_{m_{k-2}}^{0m_{k-1}0}V_{m_{k-1}}^{000}\rightarrow\left\{ m_{1},0,0,m_{2},m_{3},0,\cdots,0,m_{k-1},0,0,0,0\right\} ,\end{equation}
 where the number of different summation indices $m_{i}$ is $\left(k-1\right)$
and the length of the labeling vector is $3k.$ Since in each vector
the last three elements should always be zeros the number of different
configurations of this product is the number of ways to distribute
$\left(k-1\right)$ $m^{\prime}s$ in $3\left(k-1\right)$ cells (superscripts),
namely, $R_{k}<\binom{3\left(k-1\right)}{k-1}.$ This is only an upper
bound, since there may be some additional constraints, for example,
the first three elements in the vector should never be all zeros.
Subtracting the cases when all three first elements are zero $\binom{3\left(k-2\right)}{k-1}$
we obtain the bound\begin{equation}
R_{k}<\binom{3\left(k-1\right)}{k-1}-\binom{3\left(k-2\right)}{k-1}.\end{equation}
 This bound has the following asymptotics\begin{equation}
\lim_{k\rightarrow\infty}\frac{1}{k}\ln\left[\binom{3k-3}{k-1}-\binom{3k-6}{k-1}\right]=\ln\frac{27}{4},\end{equation}
 namely,\begin{equation}
R_{k}\leq e^{k\ln\frac{27}{4}}\leq e^{2k}.\label{eq:recur_bound}\end{equation}
 From Fig.\ref{fig:rn} we conclude that this bound is a tight bound
of the exact solution of the recurrence relation.%
\begin{figure}[tbh]
\begin{centering}
\includegraphics[width=5.0918cm,height=5.0918cm]{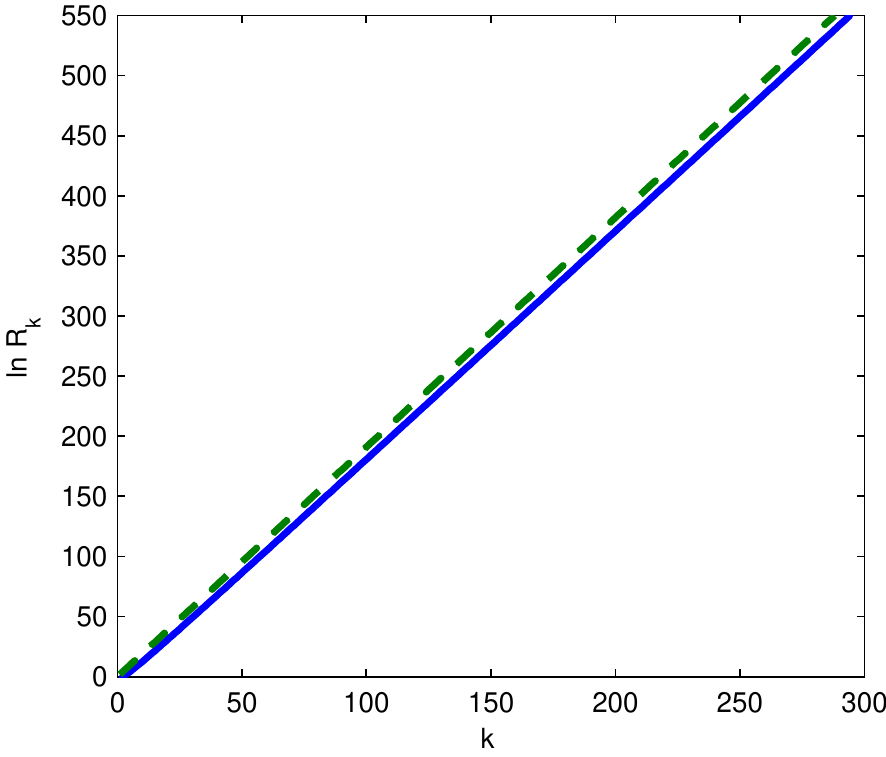}
\par\end{centering}

\caption{The solid line denotes the exact numerical solution of the recurrence
relation for $R_{k}$ and dashed line is the asymptotic upper bound
on this solution $\left(e^{2k}\right).$}

\centering{}\label{fig:rn}
\end{figure}
 This bound shows that the number of terms in the expansion increases
at most exponentially in $k$ and therefore there is no entropy problem.

\section{\label{sec:Bounding-the-general}Bounding the general term}

As clear from (\ref{eq:a2}) after the subtraction of all the secular
terms in the preceding orders the differential equation for the $k$-th
order term is\begin{align}
i\partial_{t}c_{n}^{\left(k\right)} & =-\sum_{s=0}^{k-1}E_{n}^{\left(k-s\right)}c_{n}^{\left(s\right)}+\label{eq:order_N}\\
 & +\sum_{m_{1}m_{2}m_{3}}V_{n}^{m_{1}m_{2}m_{3}}\left[\sum_{r=0}^{k-1}\sum_{s=0}^{k-1-r}\sum_{l=0}^{k-1-r-s}c_{m_{1}}^{\left(r\right)\ast}c_{m_{2}}^{\left(s\right)}c_{m_{3}}^{\left(l\right)}\right]e^{i\left(E_{n}^{\prime}+E_{m_{1}}^{\prime}-E_{m_{2}}^{\prime}-E_{m_{3}}^{\prime}\right)t}.\nonumber \end{align}
with the initial condition of $c_{n}^{\left(0\right)}=\delta_{n0}$
and the first term on the RHS is designed (see Section \ref{sec:Elimination-of-secular})
to eliminate all the time-independent part of of RHS of \eqref{eq:order_N}.
Following the construction of the lower order terms in the preceding
section by a repeated application of (\ref{eq:order_N}) the structure
of the general term in the expression for a given order $k$ can be
obtained. Note that the structure $E_{n}^{\left(l\right)}$ is similar
to the structure of $c_{n}^{\left(l\right)}$. The main blocks of
the structure take the form\begin{equation}
\zeta_{n}^{m_{1}m_{2}m_{3}}\equiv\frac{V_{n}^{m_{1}m_{2}m_{3}}}{E_{n}^{\prime}-\left\{ E^{\prime}\right\} _{m_{i}}}\label{eq:zeta}\end{equation}
where $\left\{ E^{\prime}\right\} _{m_{i}}$ denotes a sum of eigenenergies
(shifted so that the secular terms are removed, see \eqref{eq:En_expand})
that may depend on the summation indices $m_{i}.$ Then any term of
order $k$ is a product of $k$ factors of the form (\ref{eq:zeta})
and $k-1$ summations over the indices $m_{i}$\begin{equation}
\overset{k\text{ terms}}{\overbrace{\zeta_{n}^{m_{1}m_{2}m_{3}}\zeta_{m_{1}}^{m_{4}m_{5}m_{6}}\cdots\zeta_{m_{5}}^{000}\zeta_{m_{6}}^{000}}}\end{equation}
and following the last section there is an exponentially increasing
(in $k$) number of such terms. In order to bound the general term
of order $k$ we will first bound one typical block, namely,\begin{equation}
\zeta_{n}^{m_{1}m_{2}m_{3}}\equiv\frac{V_{n}^{m_{1}m_{2}m_{3}}}{E_{n}^{\prime}-\left\{ E^{\prime}\right\} _{m_{i}}}\label{eq:zeta_main_block}\end{equation}
where $\left\{ E^{\prime}\right\} _{m_{i}}$ is some sum of $E_{j}^{\prime}.$
To bound \eqref{eq:zeta_main_block} we will bound separately the
denominator and the numerator. Using Cauchy-Schwarz inequality,\begin{equation}
\left\langle \left|\zeta_{n}^{m_{1}m_{2}m_{3}}\right|^{s}\right\rangle \leq\left\langle \frac{1}{\left|E_{n}^{\prime}-\left\{ E^{\prime}\right\} _{m_{i}}\right|^{2s}}\right\rangle ^{1/2}\left\langle \left|V_{n}^{m_{1}m_{2}m_{3}}\right|^{2s}\right\rangle ^{1/2},\label{eq:zeta_cauchy_shwarz}\end{equation}
where $0<s<\frac{1}{2}$.
\begin{conjecture}
\label{con:joint_dist}For the Anderson model, which is given by the
linear part of \eqref{eq:NLSE}, the joint distribution of $R$ eigenenergies
is bounded, \[
p\left(E_{1},E_{2},\ldots,E_{R}\right)\leq\bar{D}_{R}\]
 where $\bar{D}_{R}\propto R!<\infty$.
\end{conjecture}
The conjecture is inspired by Theorem (3.1) of the recent paper by
Aizenman and Warzel \cite{Aizenman2008}. If \emph{one assumes} that
with probability one the profiles of the eigenfunctions, namely, the
squares of the eigenfunctions, which correspond to the eigenenergies
$\left\{ E_{i}\right\} _{i=1}^{R}$ are substantially different such
that $\alpha$ (as defined in Theorem (3.1) of \cite{Aizenman2008})
is bounded away from zero, than taking the intervals $I_{j}=dE_{j}$
one finds that the joint probability density can be bounded by $\bar{D}_{R}\propto\frac{R!}{\alpha^{R}}$.
It is not known how to prove that for the Anderson model the profiles
of the wave functions are distnict and how to quantify this. However,
it is reasonable to assume distinctness since different eigenfunctions
are localized in different regions and therefore are affected by different
potentials. There are double humped states (consisting of nearly symmetric
and antisymmetric combinations of two humps), which have approximately
the same squares, and therefore are natural candidates for states
that may result in violation of the conjecture. Nevertheless, those
states are very rare and the difference between their squares is exponential
in the distance between the humps. For this it is crucial that many
sites are invloved (therefore the counter example (2.1) of \cite{Aizenman2008}
is not generic). If the energies are assigned to specific locations
than the factorial term could be dropped, namely, $\bar{D}_{R}\propto\alpha^{-R}$.
This is due to the fact that specific assignment of energies chooses
one of the $R!$ permutations, mentioned in \cite{Aizenman2008}.
\begin{cor}
\label{cor:main_theorem}Given $0<s<1$, for $f={\displaystyle \sum\limits _{k=1}^{R}}c_{k}E_{i_{k}},$
where $c_{k}$ are integers (and the assignment of eigenfunctions
to sites is given by Definition \ref{def:assignment_to_sites}) the
following mean is bounded from above\begin{equation}
\left\langle \frac{1}{\left\vert f\right\vert ^{s}}\right\rangle \leq D_{R}<\infty.\end{equation}
 where $D_{R}\propto\bar{D}_{R}$.\end{cor}
\begin{proof}
By Conjecture \ref{con:joint_dist},\begin{equation}
\left\langle \frac{1}{\left\vert f\right\vert ^{s}}\right\rangle =\int\frac{p\left(E_{1},E_{2},\ldots,E_{R}\right)dE_{1}dE_{2}\cdots dE_{R}}{\left|{\displaystyle \sum\limits _{k=1}^{R}}c_{k}E_{i_{k}}\right|^{s}}\leq\bar{D}_{R}\int\frac{dE_{1}dE_{2}\cdots dE_{R}}{\left|{\displaystyle \sum\limits _{k=1}^{R}}c_{k}E_{i_{k}}\right|^{s}},\end{equation}
changing the variables to $\left\{ f,E_{2},E_{3},\ldots,E_{R}\right\} $
gives\begin{equation}
\left\langle \frac{1}{\left\vert f\right\vert ^{s}}\right\rangle \leq\bar{D}_{R}\int_{-\Delta}^{\Delta}\frac{dE_{2}\cdots dE_{R}}{\left|c_{1}\right|}\int_{f\left(\vec{E}\right)}df\frac{1}{\left|f\right|^{s}},\end{equation}
where $\left|c_{1}\right|$ is the Jacobian and $2\Delta$ is the
support of the energies. Due to the fact that $f\left(E_{1}\right)$
is linear the multiplicity is one. Since the integrand is positive
we can only increase the integral by increasing the domain of integration
of $f$. Designating by $f_{\infty}$ the maximal value of $f$,\begin{equation}
\left\langle \frac{1}{\left\vert f\right\vert ^{s}}\right\rangle \leq2\bar{D}_{R}\Delta^{R-1}\frac{f_{\infty}^{1-s}}{1-s}\equiv D_{R}.\end{equation}
\end{proof}
\begin{conjecture}
\label{con:central_limit}In the limit of $R\rightarrow\infty$, for
$0<s<1$ and for $f={\displaystyle \sum\limits _{k=1}^{R}}c_{k}E_{i_{k}},$
where $c_{k}$ are integers (and the assignment of eigenfunctions
to sites is given by Definition \ref{def:assignment_to_sites})\begin{equation}
\left\langle \frac{1}{\left\vert f\right\vert ^{s}}\right\rangle \asymp\frac{1}{R^{s/2}}\end{equation}

\end{conjecture}
For large $R$ the sum, $f={\displaystyle \sum\limits _{k=1}^{R}}c_{k}E_{i_{k}},$
can be effectively separated into groups of terms that depend on different
diagonal energies, $\varepsilon_{j}$. Therefore by the central limit
theorem, $f$ is effectively a Gaussian variable with $\left\langle f\right\rangle =0$
and $\left\langle f^{2}\right\rangle =\sigma^{2}R$, where $\sigma^{2}$
is some constant. Therefore,\begin{equation}
\left\langle \frac{1}{\left\vert f\right\vert ^{s}}\right\rangle \asymp\frac{2}{\sqrt{2\pi\sigma^{2}R}}\int_{0}^{\infty}\frac{df}{f^{s}}e^{-f^{2}/(2\sigma^{2}R)}=\frac{2}{\sqrt{2\pi}\left(\sqrt{\sigma^{2}R}\right)^{s}}\int_{0}^{\infty}\frac{df}{f^{s}}e^{-f^{2}/2}\asymp R^{-\frac{s}{2}}.\end{equation}
Conjecture \ref{con:joint_dist} and Corollary \ref{cor:main_theorem}
were tested numerically for lattice size 128, $s=\frac{1}{2}$ and
the uniform distribution\begin{equation}
\mu\left(\varepsilon\left(x\right)\right)=\begin{cases}
\frac{1}{2\Delta} & \left|\varepsilon\left(x\right)\right|\leq\Delta\\
0 & \left|\varepsilon\left(x\right)\right|>\Delta\end{cases},\label{eq:uniform_dist}\end{equation}
with $\Delta=1$. The results are presented in Fig. \ref{fig:mean_of_f}
for $R\leq10$. For $R\leq3$ all the combinations of the energies
were used and the result is an average over all the combinations.
For $R\geq4$ only a partial set of combinations of cardinality $10^{4}$,
chosen at random was used. For large $R$ the decay is as $R^{-s/2}$
in agreement with Conjecture \ref{con:central_limit}. The above calculation
was repeated for the case where the $E_{i}$ are replaced by the renormalized
energies $E_{i}'$. The calculation can be performed only to the order
$\beta^{2}$ with the help of \eqref{eq:en_oder2}. The results are
also presented in Fig. \ref{fig:mean_of_f} for $\beta=0.1$ and $\beta=1$.
\begin{conjecture}
\label{con:e_prime}Corollary \ref{cor:main_theorem} and Conjecture
\ref{con:central_limit} hold also if the $E_{i}$ are replaced by
the renormalized energies $E_{i}'$.
\end{conjecture}
The reason is that the various renormalized energies are dominated
by different independent random variables $\varepsilon_{i}$. The
numerical calculations support this point of view. In what follows
Corollary \ref{cor:main_theorem} and Conjecture \ref{con:e_prime}
(and not Conjecture \ref{con:joint_dist}) are used, and these were
tested numerically (Fig. \ref{fig:mean_of_f}).%
\begin{figure}
\includegraphics[width=0.95\textwidth]{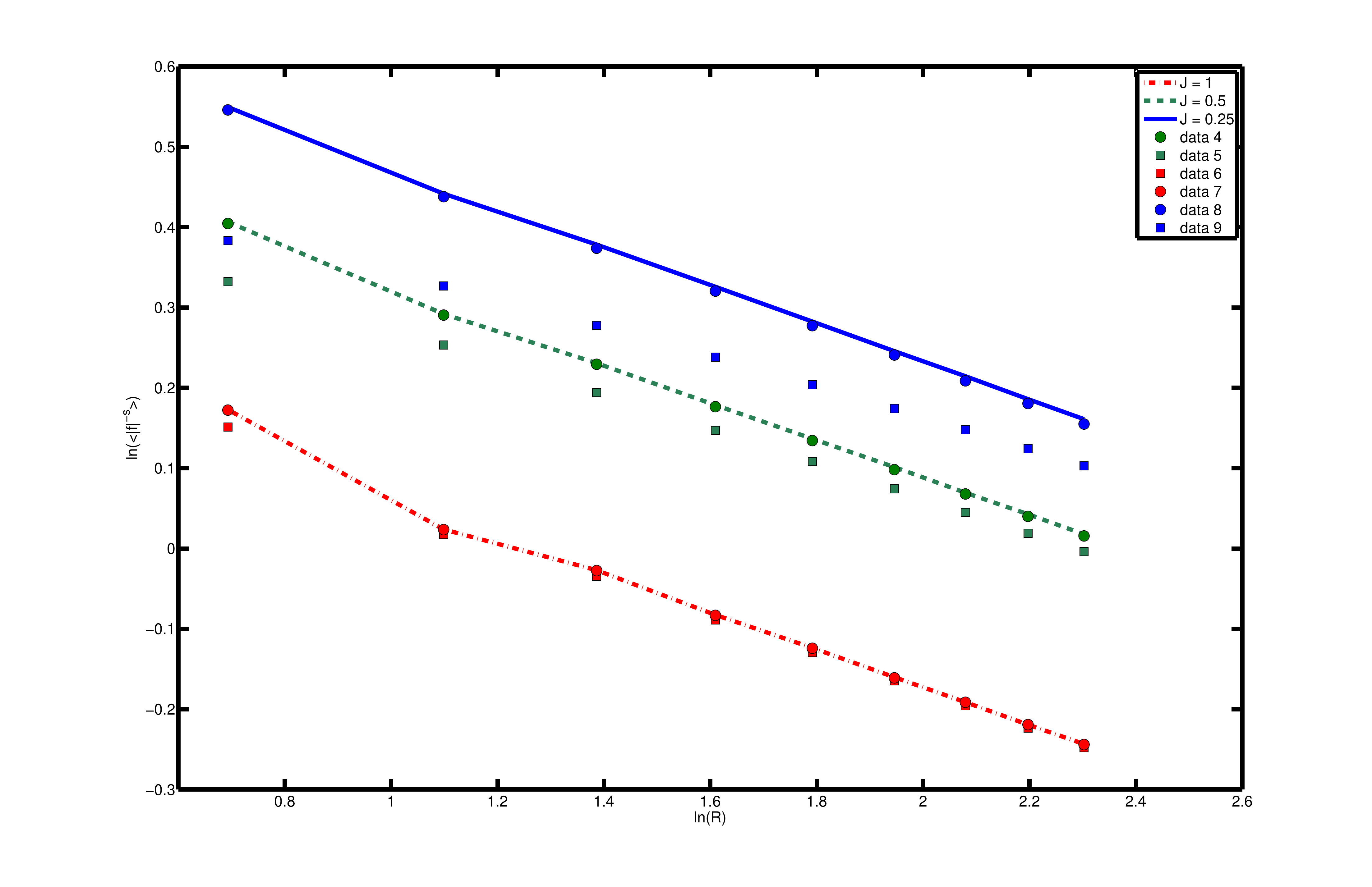}\caption{\label{fig:mean_of_f} The logarithm of $\left\langle \left|f\right|^{-1/2}\right\rangle $
as a function of the logarithm of $R$. The lines designate denominators
with $\beta=0$, with the solid line (blue) is for $J=0.25$ the dashed
line (green) is for $J=0.5$ and the dot-dashed line (red) is for
$J=1$. The solid circles and the squares are data with $\beta=1$,
and $E_{n}'$ calculated up to the second in $\beta$, such that different
colors represent different $J$, in the similar manner as for the
lines. The solid squares are for parameters similar to the ones with
the solid circles, but with the restriction that at least one of the
states that corresponds to $E_{n}'$ which is localized near the origin.}

\end{figure}

Using the bound on the overlap sum \eqref{eq:overlap} and Corollary
\ref{cor:main_theorem},\begin{equation}
\left\langle \left\vert \zeta_{n}^{m_{1}m_{2}m_{3}}\right\vert ^{s}\right\rangle _{\delta}\leq D_{\delta}\left|V_{\delta}^{\varepsilon,\varepsilon^{\prime}}\right|^{s}e^{\varepsilon s\left(\left\vert x_{n}\right\vert +\left\vert x_{m_{1}}\right\vert +\left\vert x_{m_{2}}\right\vert +\left\vert x_{m_{3}}\right\vert \right)}e^{-\frac{1}{3}\left(\gamma-\varepsilon^{\prime}\right)s\left(\left\vert x_{n}-x_{m_{1}}\right\vert +\left\vert x_{n}-x_{m_{2}}\right\vert +\left\vert x_{n}-x_{m_{3}}\right\vert \right)}.\end{equation}
this proves the proposition:
\begin{prop}
\label{pro:zeta_bound} For some $\delta,\varepsilon,\varepsilon'>0$
and $0<s<\frac{1}{2}$, \begin{equation}
\left\langle \left\vert \zeta_{n}^{m_{1}m_{2}m_{3}}\right\vert ^{s}\right\rangle _{\delta}\leq F'e^{\varepsilon s\left(\left\vert x_{n}\right\vert +\sum_{i}\left\vert x_{m_{i}}\right\vert \right)}e^{-\frac{1}{3}\left(\gamma-\varepsilon^{\prime}\right)s\sum_{i}\left\vert x_{n}-x_{m_{i}}\right\vert }\label{eq:av_zeta_bound}\end{equation}
where $F^{\prime}=D_{\delta}\left|V_{\delta}^{\varepsilon,\varepsilon^{\prime}}\right|^{s}$.
\end{prop}
Using the Chebyshev inequality, \begin{equation}
\Pr\left(\left|x\right|\geq a\right)\leq\left\langle \left|x\right|\right\rangle /a,\label{eq:Chebyshev}\end{equation}
where $x$ is a random variable and $a$ is a constant, one finds
\begin{cor}
\label{cor:zeta_bound}\begin{equation}
\Pr\left(\left\vert \zeta_{n}^{m_{1}m_{2}m_{3}}\right\vert \geq F^{\prime1/s}e^{\varepsilon\left(\left\vert x_{n}\right\vert +\sum_{i}\left\vert x_{m_{i}}\right\vert \right)}e^{-\frac{1}{3}\left(\gamma-\varepsilon^{\prime}-\eta\right)\sum_{i}\left\vert x_{n}-x_{m_{i}}\right\vert }\right)\leq e^{-\frac{\eta s}{3}\sum_{i}\left\vert x_{n}-x_{m_{i}}\right\vert }\label{eq:zeta_bound}\end{equation}
where $F^{\prime}=D_{\delta}\left|V_{\delta}^{\varepsilon,\varepsilon^{\prime}}\right|^{s}$.
\end{cor}
A general term in the expression for different orders of $c_{n}^{\left(k\right)}$
is given by the form of $\left.\left\vert \zeta_{n}^{m_{1}m_{2}m_{3}}\right\vert \left\vert \zeta_{m_{1}}^{m_{4}m_{5}m_{6}}\right\vert \cdots\left\vert \zeta_{m_{k-1}}^{000}\right\vert \right.,$
i.e., it contains $\left(k-1\right)$ summations over indices which
run over all the lattice. First we construct a general procedure to
bound a product of $k,$ $\zeta$ 's. A product of two $\zeta$ 's
is bounded by\[
\left\langle \left\vert {\displaystyle \sum\limits _{m_{1}}}\zeta_{n}^{m_{1}m_{2}m_{3}}\zeta_{m_{1}}^{m_{4}m_{5}m_{6}}\right\vert ^{s}\right\rangle _{\delta}\leq\left\langle {\displaystyle \sum\limits _{m_{1}}}\left\vert \zeta_{n}^{m_{1}m_{2}m_{3}}\right\vert ^{s}\left\vert \zeta_{m_{1}}^{m_{4}m_{5}m_{6}}\right\vert ^{s}\right\rangle _{\delta},\]
where $\left\langle \cdot\right\rangle _{\delta}$ denotes an average
over realizations where \eqref{eq:real_assump} is satisfied.

Using the Cauchy-Shwarz inequality\[
{\displaystyle \sum\limits _{m_{1}}}\left\langle \left\vert \zeta_{n}^{m_{1}m_{2}m_{3}}\right\vert ^{s}\left\vert \zeta_{m_{1}}^{m_{4}m_{5}m_{6}}\right\vert ^{s}\right\rangle _{\delta}\leq{\displaystyle \sum\limits _{m_{1}}}\left\langle \left\vert \zeta_{n}^{m_{1}m_{2}m_{3}}\right\vert ^{2s}\right\rangle _{\delta}^{1/2}\left\langle \left\vert \zeta_{m_{1}}^{m_{4}m_{5}m_{6}}\right\vert ^{2s}\right\rangle _{\delta}^{1/2}\]
setting $0<s<\frac{1}{4}$ (notice, that $s<\frac{1}{4}$ and not
$s<\frac{1}{2}$, due to \eqref{eq:zeta_cauchy_shwarz}) and inserting
the bound on the average $\left\langle \left\vert \zeta_{n}^{m_{1}m_{2}m_{3}}\right\vert ^{2s}\right\rangle _{\delta}$
from Proposition \ref{pro:zeta_bound} gives\begin{align}
 & \left\langle \left\vert {\displaystyle \sum\limits _{m_{1}}}\zeta_{n}^{m_{1}m_{2}m_{3}}\zeta_{m_{1}}^{m_{4}m_{5}m_{6}}\right\vert ^{s}\right\rangle _{\delta}\label{eq:two_zeta}\\
 & \leq F^{\prime}\exp\left[\varepsilon s\left(\left\vert x_{n}\right\vert +{\displaystyle \sum\limits _{i=2}^{6}}\left\vert x_{m_{i}}\right\vert \right)\right]e^{-s\frac{\gamma-\varepsilon'}{3}\left(\left\vert x_{n}-x_{m_{2}}\right\vert +\left\vert x_{n}-x_{m_{3}}\right\vert \right)}\times\nonumber \\
 & \times{\displaystyle \sum\limits _{m_{1}}}e^{2\varepsilon s\left\vert x_{m_{1}}\right\vert }\exp-s\frac{\gamma-\varepsilon'}{3}\left(\left\vert x_{n}-x_{m_{1}}\right\vert +{\displaystyle \sum\limits _{i=4}^{6}}\left\vert x_{m_{1}}-x_{m_{i}}\right\vert \right),\nonumber \end{align}
where we have used the inequality \[
\left(\sum_{i}\left|x_{i}\right|\right)^{s}\leq\sum_{i}\left|x_{i}\right|^{s}\qquad0<s<1.\]
Using the triangle inequality in the same manner as in \eqref{eq:first_triangle}
 \begin{align}
\left\vert x_{n}-x_{m_{1}}\right\vert +{\displaystyle \sum\limits _{i=4}^{6}}\left\vert x_{m_{1}}-x_{m_{i}}\right\vert  & \geq\frac{1}{3}{\displaystyle \sum\limits _{i=4}^{6}}\left\vert x_{n}-x_{m_{i}}\right\vert +\frac{2}{3}{\displaystyle \sum\limits _{i=4}^{6}}\left\vert x_{m_{1}}-x_{m_{i}}\right\vert \end{align}
 we get \begin{align}
 & \left\langle \left\vert {\displaystyle \sum\limits _{m_{1}}}\zeta_{n}^{m_{1}m_{2}m_{3}}\zeta_{m_{1}}^{m_{4}m_{5}m_{6}}\right\vert ^{s}\right\rangle _{\delta}\label{eq:zeta_example}\\
 & \leq F^{\prime}\exp\varepsilon s\left(\left\vert x_{n}\right\vert +{\displaystyle \sum\limits _{i=2}^{6}}\left\vert x_{m_{i}}\right\vert \right)e^{-\frac{\gamma-\varepsilon'}{3}s\left(\left\vert x_{n}-x_{m_{2}}\right\vert +\left\vert x_{n}-x_{m_{3}}\right\vert \right)}\times\nonumber \\
 & \times\exp\left[-\frac{\gamma-\varepsilon'}{9}s{\displaystyle \sum\limits _{i=4}^{6}}\left\vert x_{n}-x_{m_{i}}\right\vert \right]{\displaystyle \sum\limits _{m_{1}}}e^{2\varepsilon s\left\vert x_{m_{1}}\right\vert }e^{-2\frac{\gamma-\varepsilon'}{3}s\sum\limits _{i=4}^{6}\left\vert x_{m_{1}}-x_{m_{i}}\right\vert }\nonumber \\
 & =F''\exp\left[s\varepsilon\left(\left\vert x_{n}\right\vert +{\displaystyle \sum\limits _{i=2}^{6}}\left\vert x_{m_{i}}\right\vert \right)-\frac{\gamma-\varepsilon'}{9}s{\displaystyle \sum\limits _{i=4}^{6}}\left\vert x_{n}-x_{m_{i}}\right\vert \right]e^{-\frac{\gamma-\varepsilon'}{3}s\left(\left\vert x_{n}-x_{m_{2}}\right\vert +\left\vert x_{n}-x_{m_{3}}\right\vert \right)},\nonumber \end{align}
 where $F''\left(\gamma,\varepsilon',s,\varepsilon\right)={\displaystyle F'\sum\limits _{m_{1}}}e^{2\varepsilon s\left\vert x_{m_{1}}\right\vert }e^{-2\frac{\gamma-\varepsilon'}{3}s\sum\limits _{i=4}^{6}\left\vert x_{m_{1}}-x_{m_{i}}\right\vert }<\infty,$
in the following also other convergent sums of this type will be denoted
by $F''.$

If the term we consider is a term in the perturbation expansion it
should include some factors $\zeta_{m}^{m_{1}m_{2}m_{3}}$ with some
$m_{i}=0.$ A simple example is where $\left.x_{m_{1}}=x_{m_{2}}=x_{m_{3}}=0\right.$\begin{equation}
\left\langle \left\vert \zeta_{m}^{000}\right\vert ^{s}\right\rangle _{\delta}\leq F^{\prime}e^{-s\left(\gamma-\varepsilon'-\varepsilon\right)\left\vert x_{n}\right\vert },\label{eq:zeta_k0}\end{equation}
 where the bound was calculated using Proposition \ref{pro:zeta_bound}.
The product should terminate with a term of the form $\zeta_{m}^{000}$
therefore a term like \eqref{eq:two_zeta} is a part of a product
of the form,\begin{equation}
{\displaystyle \sum\limits _{\left\{ m_{i}\right\} }}\left\vert \zeta_{n}^{m_{1}m_{2}m_{3}}\right\vert \left\vert \zeta_{m_{1}}^{m_{4}m_{5}m_{6}}\right\vert \left\vert \zeta_{m_{2}}^{000}\right\vert \left\vert \zeta_{m_{3}}^{000}\right\vert \left\vert \zeta_{m_{4}}^{000}\right\vert \left\vert \zeta_{m_{5}}^{000}\right\vert \left\vert \zeta_{m_{6}}^{000}\right\vert .\end{equation}
 To bound it we use the generalized Hölder inequality,\[
\left\langle \prod_{i=1}^{k}\left|x_{i}\right|\right\rangle \leq\prod_{i=1}^{k}\left\langle \left|x_{i}\right|^{k}\right\rangle ^{1/k}.\]
Applying it yields,\begin{align}
 & {\displaystyle \sum\limits _{\left\{ m_{i}\right\} }}\left\langle \left\vert \zeta_{n}^{m_{1}m_{2}m_{3}}\right\vert ^{s}\left\vert \zeta_{m_{1}}^{m_{4}m_{5}m_{6}}\right\vert ^{s}\left\vert \zeta_{m_{2}}^{000}\right\vert ^{s}\left\vert \zeta_{m_{3}}^{000}\right\vert ^{s}\left\vert \zeta_{m_{4}}^{000}\right\vert ^{s}\left\vert \zeta_{m_{5}}^{000}\right\vert ^{s}\left\vert \zeta_{m_{6}}^{000}\right\vert ^{s}\right\rangle _{\delta}\\
 & \leq{\displaystyle \sum\limits _{\left\{ m_{i}\right\} }}\left(\left\langle \left\vert \zeta_{n}^{m_{1}m_{2}m_{3}}\right\vert ^{7s}\right\rangle _{\delta}\left\langle \left\vert \zeta_{m_{1}}^{m_{4}m_{5}m_{6}}\right\vert ^{7s}\right\rangle _{\delta}{\displaystyle \prod\limits _{i=2}^{6}}\left\langle \left\vert \zeta_{m_{i}}^{000}\right\vert ^{7s}\right\rangle _{\delta}\right)^{1/7}\nonumber \end{align}
 setting $0<s<\frac{1}{14}$ and inserting the bounds on the averages
\eqref{eq:two_zeta} and \eqref{eq:zeta_k0} gives\begin{align}
 & \left\langle \left({\displaystyle \sum\limits _{\left\{ m_{i}\right\} }}\zeta_{n}^{m_{1}m_{2}m_{3}}\zeta_{m_{1}}^{m_{4}m_{5}m_{6}}\zeta_{m_{2}}^{000}\zeta_{m_{3}}^{000}\zeta_{m_{4}}^{000}\zeta_{m_{5}}^{000}\zeta_{m_{6}}^{000}\right)^{s}\right\rangle _{\delta}\label{eq:zeta_example1}\\
 & \leq F^{\prime}{\displaystyle \sum\limits _{\left\{ m_{i}\right\} }}\exp\varepsilon s\left(\left\vert x_{n}\right\vert +{\displaystyle \sum\limits _{i=2}^{6}}\left\vert x_{m_{i}}\right\vert \right)e^{-\frac{\gamma-\varepsilon'}{3}s\left(\left\vert x_{n}-x_{m_{2}}\right\vert +\left\vert x_{n}-x_{m_{3}}\right\vert \right)}\times\nonumber \\
 & \times\exp\left[-\frac{\gamma-\varepsilon'}{9}s{\displaystyle \sum\limits _{i=4}^{6}}\left\vert x_{n}-x_{m_{i}}\right\vert -\left(\gamma-\varepsilon'-\varepsilon\right)s{\displaystyle \sum\limits _{i=2}^{6}}\left\vert x_{m_{i}}\right\vert \right]\nonumber \\
 & =F^{\prime}e^{\varepsilon s\left\vert x_{n}\right\vert }{\displaystyle \sum\limits _{\left\{ m_{i}\right\} }}e^{-\frac{\gamma-\varepsilon'}{3}s\left(\left\vert x_{n}-x_{m_{2}}\right\vert +\left\vert x_{n}-x_{m_{3}}\right\vert \right)}\times\nonumber \\
 & \times\exp\left[-\frac{\gamma-\varepsilon'}{9}s{\displaystyle \sum\limits _{i=4}^{6}}\left\vert x_{n}-x_{m_{i}}\right\vert -\left(\gamma-\varepsilon'-2\varepsilon\right)s{\displaystyle \sum\limits _{i=2}^{6}}\left\vert x_{m_{i}}\right\vert \right],\nonumber \end{align}
 where $\left\{ m_{i}\right\} $ stands for a sum over all the $m_{i}.$
Using the inequality,\begin{equation}
{\displaystyle \sum\limits _{i=4}^{6}}\left(\left\vert x_{n}-x_{m_{i}}\right\vert +\left\vert x_{m_{i}}\right\vert \right)\geq3\left\vert x_{n}\right\vert ,\end{equation}
 we get\begin{align}
 & \left\langle \left({\displaystyle \sum\limits _{\left\{ m_{i}\right\} }}\zeta_{n}^{m_{1}m_{2}m_{3}}\zeta_{m_{1}}^{m_{4}m_{5}m_{6}}\zeta_{m_{2}}^{000}\zeta_{m_{3}}^{000}\zeta_{m_{4}}^{000}\zeta_{m_{5}}^{000}\zeta_{m_{6}}^{000}\right)^{s}\right\rangle _{\delta}\label{eq:zeta_example2}\\
 & \leq F^{\prime}e^{-\left(\gamma-\varepsilon'-\varepsilon\right)s\left\vert x_{n}\right\vert }{\displaystyle \sum\limits _{\left\{ m_{i}\right\} }}\exp-2s\left(\frac{\gamma-\varepsilon'}{3}-\varepsilon\right){\displaystyle \sum\limits _{i=2}^{6}}\left\vert x_{m_{i}}\right\vert \nonumber \\
 & =F^{\prime}e^{-\left(\gamma-\varepsilon'-\varepsilon\right)s\left\vert x_{n}\right\vert }.\nonumber \end{align}
 Let us study the form of a general graph (c.f. Fig. \ref{fig:tree}).
\begin{figure}
\begin{centering}
\includegraphics[clip,width=0.75\textwidth]{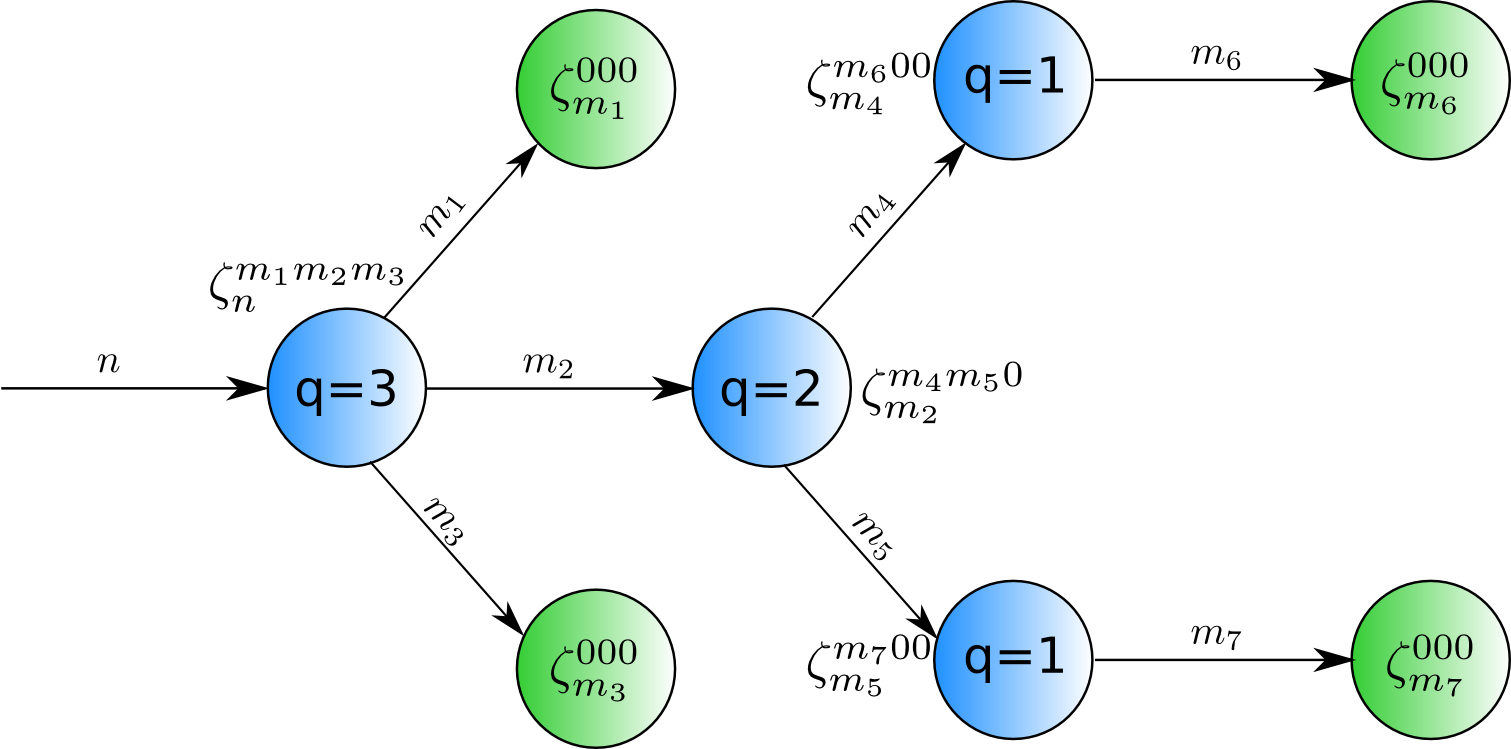}
\par\end{centering}

\caption{An example of a graph that is used to construct the general term.
The graph describes an 8-th order term, $\zeta_{n}^{m_{1}m_{2}m_{3}}\zeta_{m_{1}}^{000}\zeta_{m_{3}}^{000}\zeta_{m_{2}}^{m_{4}m_{5}0}\zeta_{m_{4}}^{m_{6}00}\zeta_{m_{5}}^{m_{7}00}\zeta_{m_{6}}^{000}\zeta_{m_{7}}^{000}.$}

\centering{}\label{fig:tree}
\end{figure}

It can be described as a tree starting from $n,$ the \textquotedbl{}root\textquotedbl{}
and four types of branching points, where $q$ $\left(0\leq q\leq3\right)$
branches continue while $3-q$ branches terminate. A branch which
continues is associated with a value $m_{i}\neq0,$ while a branch
terminates with $m_{i}=0.$ In the above, bounds on branches with
$q=0$ and 3 are calculated, and are given by \eqref{eq:av_zeta_bound}
and \eqref{eq:zeta_k0}, respectively. The bounds for $q=1,2$ follow
similarly from Proposition \ref{pro:zeta_bound}. Along each {}``bond''
from the \textquotedbl{}root\textquotedbl{} to the {}``leaves''
a term $\zeta_{n}^{m_{1}m_{2}m_{3}}$ is multiplied. At a point from
where $q$ branches continue the exponent is reduced by a factor of
$q.$ In other words, with $x_{n}$ and $x_{m}$ are connected by
a path that is crossing $l$ branching points with ratios $q_{1},\ldots,q_{l}$
the bound on the product of the zetas contains a factor $\exp-\frac{\gamma}{\tilde{q}}\left\vert x_{n}-x_{m}\right\vert $
where $\tilde{q}={\displaystyle \prod\limits _{i=1}^{l}}q_{i}.$ The
total number of branch ends is also $\tilde{q}.$ To terminate all
the branches $\tilde{q}$ factors of the form $\zeta_{m_{i}}^{000}$
(bounded by (\ref{eq:zeta_k0})) should be multiplied resulting in
the a term that multiplies the bound on a sum of the form\begin{align}
 & {\displaystyle \sum\limits _{\left\{ m_{i}\right\} }}\exp\left[-\frac{\left(\gamma-\varepsilon'\right)}{\tilde{q}}s{\displaystyle \sum\limits _{i}}\left\vert x_{n}-x_{m_{i}}\right\vert -\left(\gamma-\varepsilon'\right)s{\displaystyle \sum\limits _{i}}\left\vert x_{m_{i}}\right\vert \right]\\
 & \leq\exp-\frac{\left(\gamma-\varepsilon'\right)}{\tilde{q}}s{\displaystyle \sum\limits _{i}}\left\vert x_{n}\right\vert {\displaystyle \sum\limits _{\left\{ m_{i}\right\} }}\exp-\left(1-\frac{1}{\tilde{q}}\right)\left(\gamma-\varepsilon'\right)s{\displaystyle \sum\limits _{i}}\left\vert x_{m_{i}}\right\vert \nonumber \\
 & =e^{-\frac{\gamma-\varepsilon'}{\tilde{q}}s\tilde{q}\left\vert x_{n}\right\vert }{\displaystyle \sum\limits _{\left\{ m_{i}\right\} }}\exp-\left(1-\frac{1}{\tilde{q}}\right)\left(\gamma-\varepsilon'\right)s{\displaystyle \sum\limits _{i}}\left\vert x_{m_{i}}\right\vert \equiv S_{\varepsilon}e^{-\left(\gamma-\varepsilon'\right)s\left\vert x_{n}\right\vert }\nonumber \end{align}
 restoring the original convergence rate. As the product consists
of $k$ terms the evaluation of $\left\vert \zeta_{n}^{m_{1}m_{2}m_{3}}\right\vert ^{sk}$
is required for the use of the Hölder inequality. Therefore it is
required that $0<s<1/2k.$
\begin{lem}
\label{lem:main_lemma}For a given $k$ (the number of $\zeta$'s),
$\delta,\varepsilon,\varepsilon',\eta'>0$ and $0<s<\frac{1}{2k}$
\begin{equation}
\left\langle \left\vert {\displaystyle \sum\limits _{\left\{ m_{i}\right\} }}\zeta_{n}^{m_{1}m_{2}m_{3}}\zeta_{m_{1}}^{m_{4}m_{5}m_{6}}\cdots\zeta_{m_{N-1}}^{000}\right\vert ^{s}\right\rangle _{\delta}\leq F_{\delta}^{\left(k\right)}e^{-\left(\gamma-\varepsilon-\varepsilon'\right)s\left\vert x_{n}\right\vert }\end{equation}
 or using the Chebyshev inequality \eqref{eq:Chebyshev}\begin{equation}
\Pr\left(\left\vert {\displaystyle \sum\limits _{\left\{ m_{i}\right\} }}\zeta_{n}^{m_{1}m_{2}m_{3}}\zeta_{m_{1}}^{m_{4}m_{5}m_{6}}\cdots\zeta_{m_{N-1}}^{000}\right\vert \geq\left(F_{\delta}^{\left(k\right)}\right)^{1/s}e^{-\left(\gamma-\varepsilon-\varepsilon'-\eta^{\prime}\right)\left\vert x_{n}\right\vert }\right)\leq e^{-\eta^{\prime}s\left\vert x_{n}\right\vert },\end{equation}
 where $F_{\delta}^{\left(k\right)}$ is a constant which is built
iteratively by the construction demonstrated in \eqref{eq:zeta_example1}
and \eqref{eq:zeta_example2} and is proportional to $D_{\delta}$.
\end{lem}
It is of importance, that any product with the same number of zetas
has the same bound with the same probability. This allows us to bound
$c_{n}^{\left(k\right)}$ by counting the number of different configurations,
$R_{k}$, of the product for a given $k$ and then multiplying it
by the bound of each product. This proves the theorem:
\begin{thm}
\label{thm:main_theorem}For a given $k$ and $\delta,\varepsilon,\varepsilon',\eta'>0$\begin{equation}
\Pr\left(\left\vert c_{n}^{\left(k\right)}\right\vert \geq\left(F_{\delta}^{\left(k\right)}\right)^{k}e^{ck^{2}+c^{\prime}k}e^{-\left(\gamma-\varepsilon-\varepsilon'-\eta^{\prime}\right)\left\vert x_{n}\right\vert }\right)\leq e^{-c^{\prime}}e^{-\eta^{\prime}\left\vert x_{n}\right\vert /k}.\label{eq:theorem1}\end{equation}
where $F_{\delta}^{\left(k\right)}$ which is proportional to $D_{\delta}$
and $c$ and $c'$ are constants.\end{thm}
\begin{proof}
Using Lemma \ref{lem:main_lemma} and summing over configurations
denoted by $i_{c}$ one obtains the bound \begin{align}
\left\langle \left\vert c_{n}^{\left(k\right)}\right\vert ^{s}\right\rangle _{\delta} & =\left\langle \left\vert \sum_{i_{c}=1}^{R_{k}}\left\{ {\displaystyle \sum\limits _{\left\{ m_{i}\right\} }}\zeta_{n}^{m_{1}m_{2}m_{3}}\zeta_{m_{1}}^{m_{4}m_{5}m_{6}}\cdots\zeta_{m_{k-1}}^{000}\right\} \right\vert ^{s}\right\rangle _{\delta}\\
 & \leq\sum_{i_{c}=1}^{R_{k}}\left\langle \left\vert {\displaystyle \sum\limits _{\left\{ m_{i}\right\} }}\zeta_{n}^{m_{1}m_{2}m_{3}}\zeta_{m_{1}}^{m_{4}m_{5}m_{6}}\cdots\zeta_{m_{k-1}}^{000}\right\vert ^{s}\right\rangle _{\delta}\leq R_{k}F_{\delta}^{\left(k\right)}e^{-\left(\gamma-\varepsilon-\varepsilon'\right)s\left\vert x_{n}\right\vert }.\nonumber \end{align}
Using \eqref{eq:recur_bound}\begin{equation}
\left\langle \left\vert c_{n}^{\left(k\right)}\right\vert ^{s}\right\rangle _{\delta}\leq e^{2k}F_{\delta}^{\left(k\right)}e^{-\left(\gamma-\varepsilon-\varepsilon'\right)s\left\vert x_{n}\right\vert }\end{equation}
 or \begin{equation}
\Pr\left(\left\vert c_{n}^{\left(k\right)}\right\vert \geq e^{c^{\prime}/s}\left(F_{\delta}^{\left(k\right)}\right)^{1/s}e^{-\left(\gamma-\varepsilon-\varepsilon'-\eta^{\prime}\right)\left\vert x_{n}\right\vert }e^{2k/s}\right)\leq e^{-c^{\prime}}e^{-\eta^{\prime}s\left\vert x_{n}\right\vert }\end{equation}
Choosing the largest $s$ produces a bound where $c,c'$ are constants
and $F_{\delta}^{\left(k\right)}$ is a constant which is built iteratively
by the construction demonstrated in \eqref{eq:zeta_example1} and
\eqref{eq:zeta_example2} and is proportional to $D_{\delta}$. \end{proof}
\begin{rem}
From \eqref{eq:recur_bound} one sees that $c\backsimeq2$ and later
we set $c'=cN$.
\end{rem}

\section{\label{sec:Numerical-results}Numerical results}

In this section we will check how well the perturbation series up
to the second order in $\beta$ approximates the numerical solution
of \eqref{eq:NLSE}. For this purpose we use the expressions for $c_{n}^{\left(1\right)}$
and $c_{n}^{\left(2\right)}$ which were obtained in \eqref{eq:cn1}
and \eqref{eq:cn_order2}, respectively, and also the expression for
the renormalized energies $E_{n}'$ up to second order in $\beta$
which are given by \eqref{eq:en_oder2}. We use the perturbation expansion
up to the second order in $\beta$, namely\begin{equation}
\bar{c}_{n}=c_{n}^{\left(0\right)}+\beta c_{n}^{\left(1\right)}+\beta^{2}\left.c_{n}^{\left(2\right)}\right|_{\beta=0},\label{eq:second_order_numerics}\end{equation}
 where we took $\left.c_{n}^{\left(2\right)}\right|_{\beta=0}$ in
order to keep only the contribution of $c_{n}^{\left(2\right)}$ up
to the order $\beta^{2}$. To compare, we plot the real and the imaginary
parts of both the numerical solution of \eqref{eq:NLSE} with the
distribution \eqref{eq:uniform_dist} and the perturbative approximation
$\bar{c}_{n}$. From figures \ref{fig:central_beta_0.1} and \ref{fig:central_beta_0.01}
we see that the correspondence between the numerical solution of \eqref{eq:NLSE}
and the perturbative approximation is good for times $<50$ for $\beta=0.1$
and times $<200$ for $\beta=0.01$. %
\begin{figure}
\includegraphics[width=0.95\textwidth]{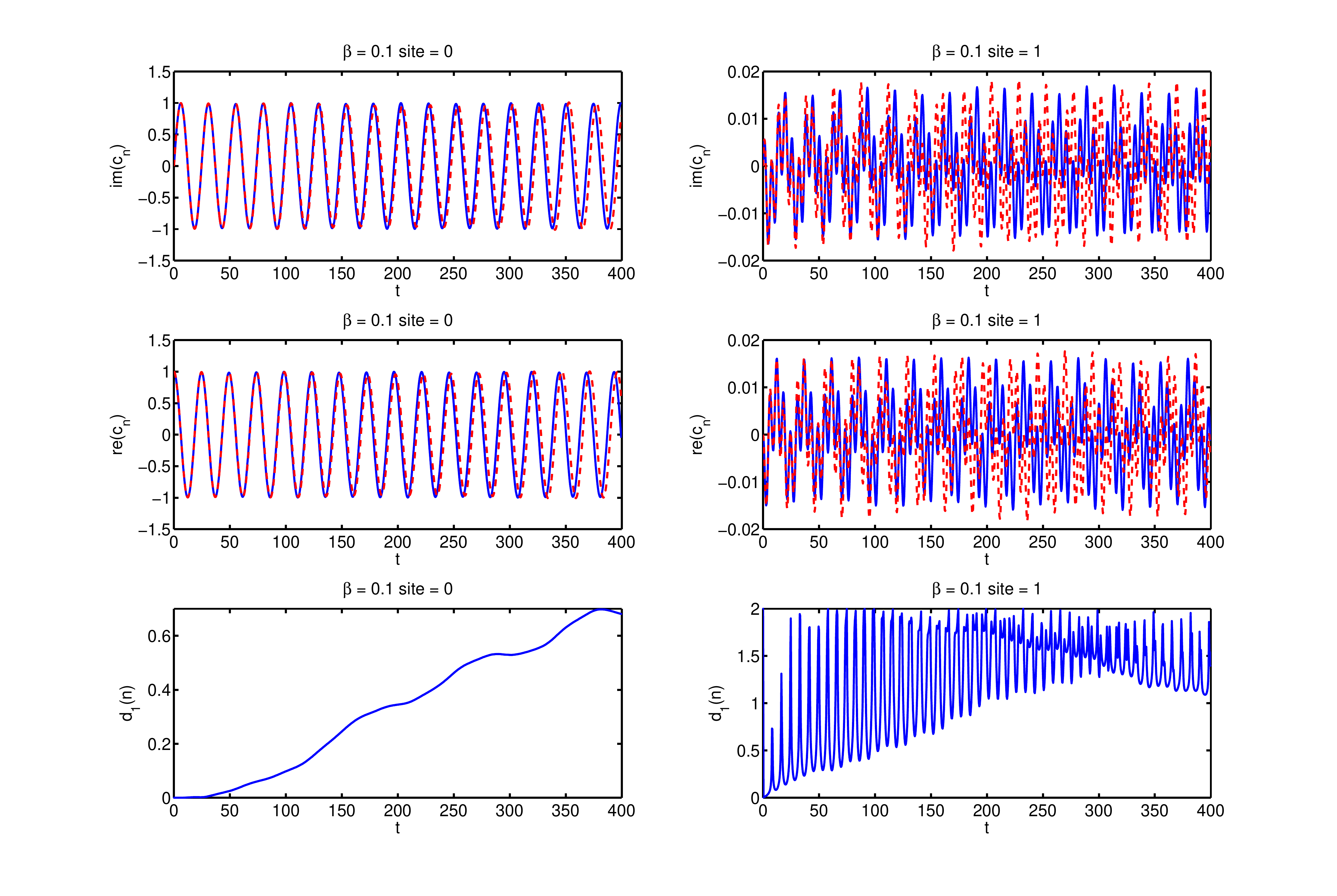}\caption{\label{fig:central_beta_0.1}In the first and the second rows of the
figure are presented the imaginary and real parts, respectively, of
the numerical solution of \eqref{eq:NLSE} (blue, solid) and the perturbative
approximation (red, dashed) as function of time for $c_{0}$ (left
column) and $c_{1}$ (right column). In the third row the relative
difference, $d_{1}\left(t\right),$ as defined in Definition \ref{def:The-relative-difference}
is plotted. The parameters of the plot are, $J=0.25$, $\Delta=1$
and $\beta=0.1$ and lattice size of 128. See \eqref{eq:NLSE} for
the definition of the constants.}

\end{figure}
\begin{figure}
\includegraphics[width=0.95\textwidth]{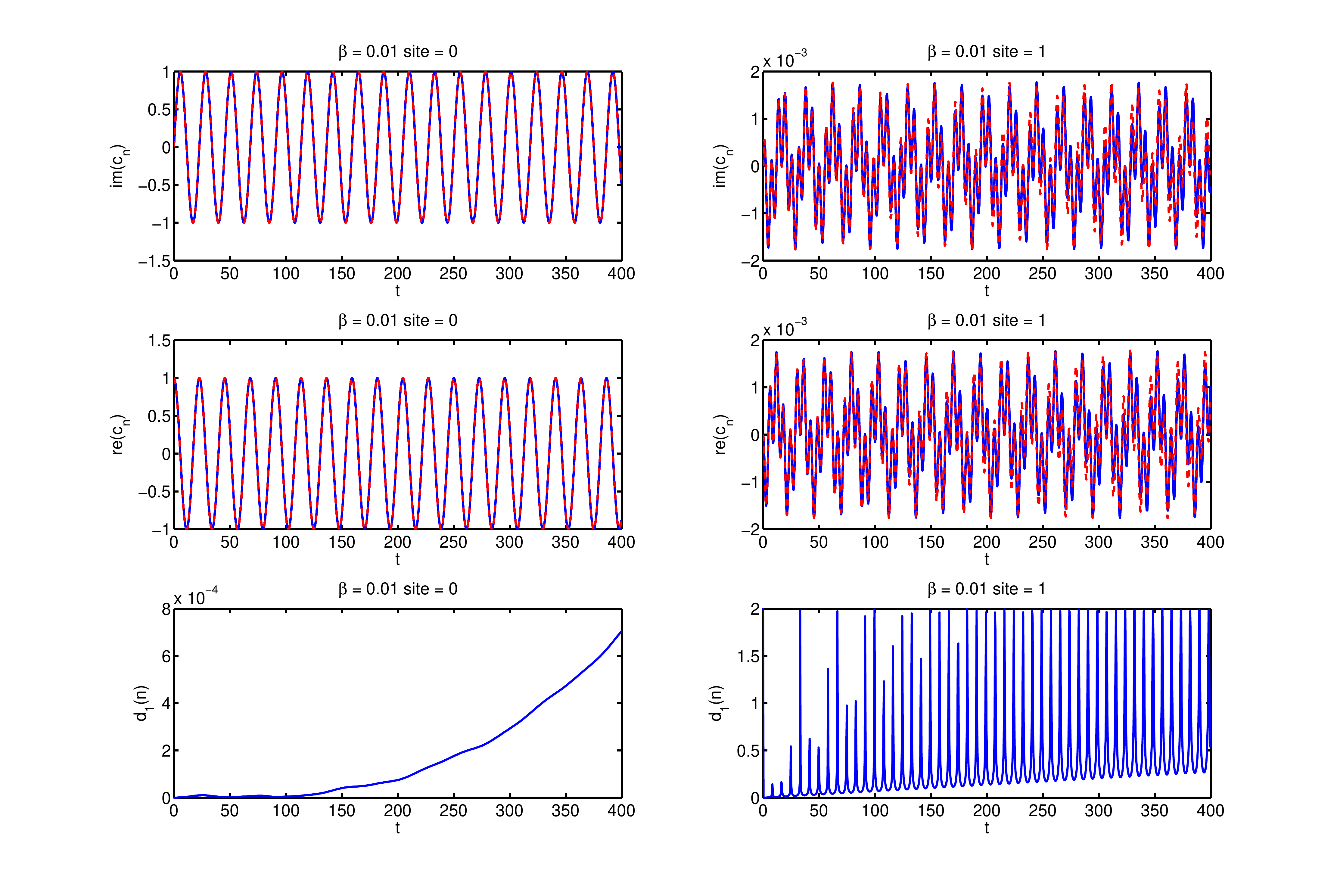}\caption{\label{fig:central_beta_0.01}Same as Figure \ref{fig:central_beta_0.1}
but for $\beta=0.01$.}

\end{figure}
 Additionally, the correspondence of the central site, $c_{0}$, which
is used as the initial condition, $c_{n}\left(t=0\right)=\delta_{n0}$,
is much better than the correspondence of the neighboring sites. A
possible explanation for this could be that the nonlinear perturbation
is more pronounced at the states $n$ with $c_{n}\left(t=0\right)=0$,
this is due to the fact that for $\beta\mbox{=0}$ those sites are
unpopulated (zero) for all times, resulting in lower signal to noise
ratio. To examine the convergence in time we will define a time, $t^{*}$.
For times $t<t^{*}$ the relative difference between the $L_{2}$
norms of the exact and the approximate solutions, $d_{2}\left(t\right)$,
defined bellow, is less than 10\%. It is instructive to introduce
the following definitions.
\begin{defn}
\label{def:The-relative-difference}The relative difference between
the exact and the perturbative solution at site $n$ is defined as\[
d_{1}\left(t\right)=\frac{2\left|c_{n}\left(t\right)-\bar{c}_{n}\left(t\right)\right|}{\left|c_{n}\left(t\right)\right|+\left|\bar{c}_{n}\left(t\right)\right|}.\]

\end{defn}
~
\begin{defn}
\label{def:t_star}$t^{*}$ is a time until which the relative difference
between the exact and the perturbative solution, $d_{2}\left(t\right)$,
is less than 10\%.
\end{defn}
\begin{equation}
d_{2}\left(t\right)=\frac{\sum_{n}\left|c_{n}\left(t\right)-\bar{c}_{n}\left(t\right)\right|^{2}}{\sum_{n}\left|c_{n}\left(t\right)\right|^{2}}\leq0.1.\end{equation}
In figure \ref{fig:t_star} we see that as $\beta$ becomes smaller
the time for which the expansion to second order in $\beta$ is close
to the exact solution (within 10\%), $t^{*}$, increases.%
\begin{figure}[H]
\includegraphics[width=0.75\textwidth]{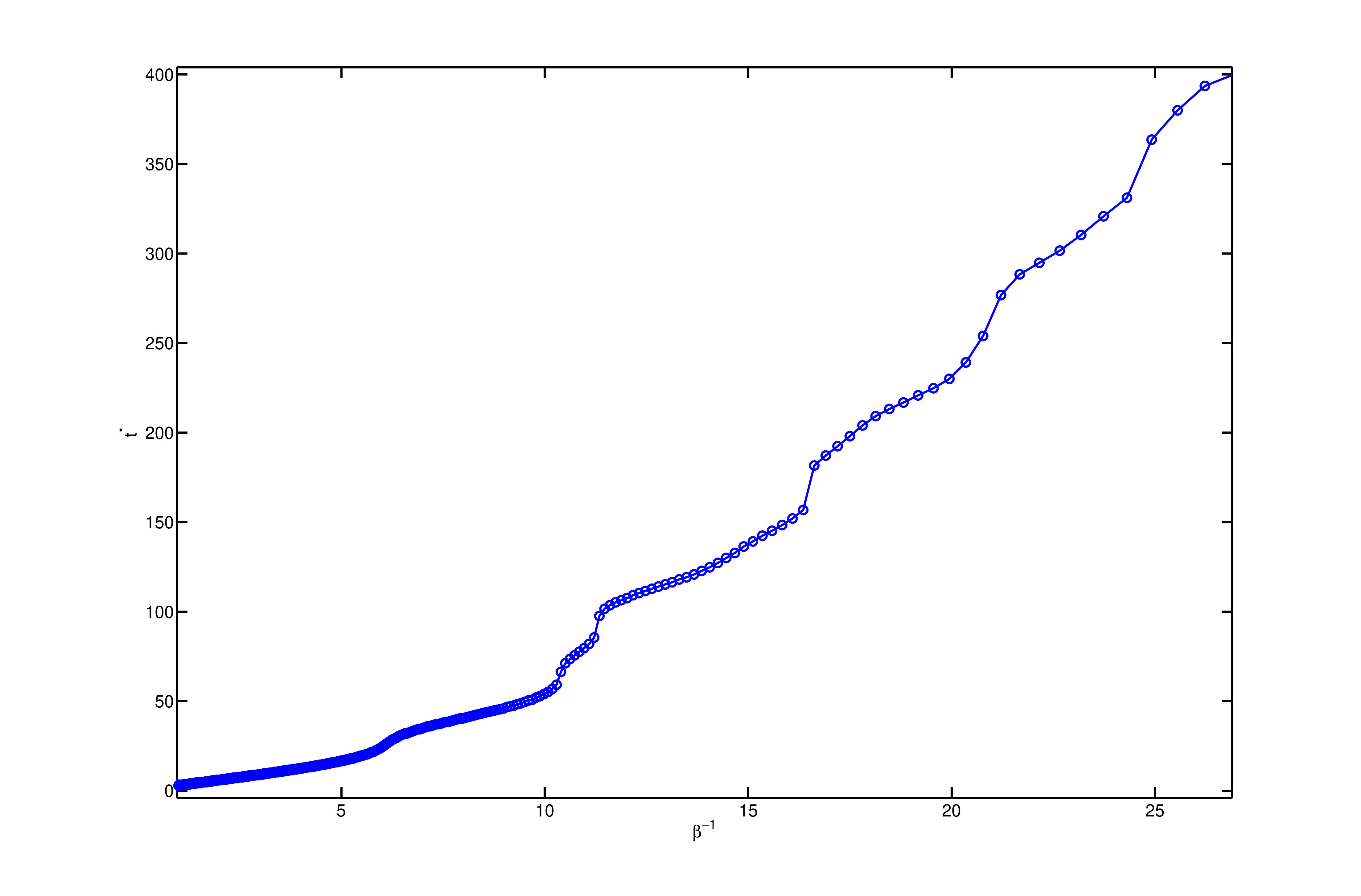}\caption{\label{fig:t_star}$t^{*}$ (see text) as a function of the inverse
nonlinearity strength, $\beta^{-1}$.}

\end{figure}

\section{\label{sec:Bounding-the-remainder}Bounding the remainder}

In order to bound the solution we have to bound, $Q_{n}$, the remainder
of the expansion (\ref{eq:cn_expand})\begin{equation}
c_{n}\left(t\right)=c_{n}^{\left(0\right)}+\beta c_{n}^{\left(1\right)}+\beta^{2}c_{n}^{\left(2\right)}+\cdots+\beta^{N-1}c_{n}^{\left(N-1\right)}+\beta^{N}Q_{n}.\end{equation}
This is achieved applying the bootstrap argument to the remainder.
Substituting the expansion \eqref{eq:cn_expand} into the equation
\eqref{eq:diff_eq} and writing a differential equation for the remainder
gives\begin{align}
i\partial_{t}Q_{n} & ={\displaystyle \sum\limits _{m_{1}m_{2}m_{3}}}\sum_{N-1\leq k+l+q\leq3\left(N-1\right)}^{N-1}\beta^{k+l+q+1-N}V_{n}^{m_{1}m_{2}m_{3}}c_{m_{1}}^{\ast\left(q\right)}c_{m_{2}}^{\left(l\right)}c_{m_{3}}^{\left(k\right)}e^{i\left(E_{n}^{\prime}+E_{m_{1}}^{\prime}-E_{m_{2}}^{\prime}-E_{m_{3}}^{\prime}\right)t}\\
 & +{\displaystyle \sum\limits _{m_{1}m_{2}m_{3}}}\sum_{0\leq k+l\leq2\left(N-1\right)}^{N-1}\beta^{l+k+1}V_{n}^{m_{1}m_{2}m_{3}}Q_{m_{1}}^{\ast}c_{m_{2}}^{\left(l\right)}c_{m_{3}}^{\left(k\right)}e^{i\left(E_{n}^{\prime}+E_{m_{1}}^{\prime}-E_{m_{2}}^{\prime}-E_{m_{3}}^{\prime}\right)t}\nonumber \\
 & +2{\displaystyle \sum\limits _{m_{1}m_{2}m_{3}}}\sum_{0\leq k+l\leq2\left(N-1\right)}^{N-1}\beta^{l+k+1}V_{n}^{m_{1}m_{2}m_{3}}Q_{m_{3}}c_{m_{1}}^{\left(l\right)\ast}c_{m_{2}}^{\left(k\right)}e^{i\left(E_{n}^{\prime}+E_{m_{1}}^{\prime}-E_{m_{2}}^{\prime}-E_{m_{3}}^{\prime}\right)t}\nonumber \\
 & +2{\displaystyle \sum\limits _{m_{1}m_{2}m_{3}}}\sum_{k=0}^{N-1}\beta^{N+k+1}V_{n}^{m_{1}m_{2}m_{3}}Q_{m_{1}}^{\ast}Q_{m_{2}}c_{m_{3}}^{\left(k\right)}e^{i\left(E_{n}^{\prime}+E_{m_{1}}^{\prime}-E_{m_{2}}^{\prime}-E_{m_{3}}^{\prime}\right)t}\nonumber \\
 & +{\displaystyle \sum\limits _{m_{1}m_{2}m_{3}}}\sum_{k=0}^{N-1}\beta^{N+k+1}V_{n}^{m_{1}m_{2}m_{3}}c_{m_{1}}^{\left(k\right)\ast}Q_{m_{2}}Q_{m_{3}}e^{i\left(E_{n}^{\prime}+E_{m_{1}}^{\prime}-E_{m_{2}}^{\prime}-E_{m_{3}}^{\prime}\right)t}\nonumber \\
 & +\sum_{m_{1}m_{2}m_{3}}\beta^{2N+1}V_{n}^{m_{1}m_{2}m_{3}}Q_{m_{1}}^{\ast}Q_{m_{2}}Q_{m_{3}}e^{i\left(E_{n}^{\prime}+E_{m_{1}}^{\prime}-E_{m_{2}}^{\prime}-E_{m_{3}}^{\prime}\right)t}\nonumber \end{align}
where the sums over orders are understood as follows. By $\sum_{N-1\leq k+l+q\leq3\left(N-1\right)}^{N-1}$
we mean $\sum_{k,l,q=0}^{N-1}$ with the constraint $\left.N-1\leq k+l+q\right..$
Integrating and using the fact, $Q_{n}\left(t=0\right)=0$,\begin{align}
\frac{\left\vert Q_{n}\right\vert }{t} & \leq{\displaystyle \sum\limits _{m_{1}m_{2}m_{3}}}\sum_{N-1\leq k+l+q\leq3\left(N-1\right)}^{N-1}\beta^{k+l+q+1-N}\left\vert V_{n}^{m_{1}m_{2}m_{3}}\right\vert \sup_{0\leq t'\leq t}\left[\left\vert c_{m_{1}}^{\ast\left(k\right)}\right\vert \left\vert c_{m_{2}}^{\left(l\right)}\right\vert \left\vert c_{m_{3}}^{\left(q\right)}\right\vert \right]\label{eq:bootstrap1}\\
 & +{\displaystyle \sum\limits _{m_{1}m_{2}m_{3}}}\sum_{0\leq k+l\leq2\left(N-1\right)}^{N-1}\beta^{l+k+1}\left\vert V_{n}^{m_{1}m_{2}m_{3}}\right\vert \sup_{0\leq t'\leq t}\left[\left\vert Q_{m_{1}}^{\ast}\right\vert \left\vert c_{m_{2}}^{\left(l\right)}\right\vert \left\vert c_{m_{3}}^{\left(k\right)}\right\vert \right]\nonumber \\
 & +2{\displaystyle \sum\limits _{m_{1}m_{2}m_{3}}}\sum_{0\leq k+l\leq2\left(N-1\right)}^{N-1}\beta^{l+k+1}\left\vert V_{n}^{m_{1}m_{2}m_{3}}\right\vert \sup_{0\leq t'\leq t}\left[\left\vert Q_{m_{3}}\right\vert \left\vert c_{m_{1}}^{\left(l\right)\ast}\right\vert \left\vert c_{m_{2}}^{\left(k\right)}\right\vert \right]\nonumber \\
 & +2{\displaystyle \sum\limits _{m_{1}m_{2}m_{3}}}\sum_{k=0}^{N-1}\beta^{N+k+1}\left\vert V_{n}^{m_{1}m_{2}m_{3}}\right\vert \sup_{0\leq t'\leq t}\left[\left\vert Q_{m_{1}}^{\ast}\right\vert \left\vert Q_{m_{2}}\right\vert \left\vert c_{m_{3}}^{\left(k\right)}\right\vert \right]\nonumber \\
 & +{\displaystyle \sum\limits _{m_{1}m_{2}m_{3}}}\sum_{k=0}^{N-1}\beta^{N+k+1}\left\vert V_{n}^{m_{1}m_{2}m_{3}}\right\vert \sup_{0\leq t'\leq t}\left[\left\vert Q_{m_{2}}\right\vert \left\vert Q_{m_{3}}\right\vert \left\vert c_{m_{1}}^{\left(k\right)\ast}\right\vert \right]\nonumber \\
 & +\beta^{2N+1}\sum_{m_{1}m_{2}m_{3}}\left\vert V_{n}^{m_{1}m_{2}m_{3}}\right\vert \sup_{0\leq t'\leq t}\left[\left\vert Q_{m_{1}}^{\ast}\right\vert \left\vert Q_{m_{2}}\right\vert \left\vert Q_{m_{3}}\right\vert \right],\nonumber \end{align}
where $t$ is the upper limit of the integration.

Since $Q_{n}\left(t\right)$ is continuous and $Q_{n}\left(t=0\right)=0$,
for small $t$, $Q_{n}\left(t\right)\sim t\cdot S_{1}$. Therefore
we can always find a sufficiently small $\tau>0$, such that \begin{equation}
\left|Q_{n}\left(\tau\right)\right|<2\tau\cdot S_{1},\label{eq:bootstrap1-1}\end{equation}
 where \begin{equation}
S_{1}={\displaystyle \sum\limits _{m_{1}m_{2}m_{3}}}\sum_{N-1\leq k+l+q\leq3\left(N-1\right)}^{N-1}\beta^{k+l+q+1-N}\left\vert V_{n}^{m_{1}m_{2}m_{3}}\right\vert \sup_{0\leq t'\leq t}\left[\left\vert c_{m_{1}}^{\ast\left(k\right)}\right\vert \left\vert c_{m_{2}}^{\left(l\right)}\right\vert \left\vert c_{m_{3}}^{\left(q\right)}\right\vert \right].\end{equation}
Assume, that there is some time, $t$, for which,\begin{equation}
\left|Q_{n}\left(t\right)\right|>2t\cdot S_{1},\label{eq:bootstrap1-2}\end{equation}
than since $Q_{n}\left(t\right)$ is continuous and \eqref{eq:bootstrap1-1}
holds there is a time, $t_{0}$, where \begin{equation}
\left|Q_{n}\left(t_{0}\right)\right|=2t_{0}\cdot S_{1}.\label{eq:t0_definition}\end{equation}
Inserting this equality into the inequality \eqref{eq:bootstrap1},
we get an interval $0\leq t\leq t_{0}$ for which \eqref{eq:bootstrap1-1}
holds (see \eqref{eq:bootstrap_t0}). We proceed by bounding $S_{1}$
and other sums of \eqref{eq:bootstrap1}.

Using the Theorem \ref{thm:main_theorem} obtained in the end of the
Section \ref{sec:Bounding-the-general}, we can bound, $S_{1}$. The
inequality is violated with a probability found from \eqref{eq:theorem1}.
We start with the bound\begin{align}
S_{1} & \leq{\displaystyle \sum\limits _{\left\{ m_{i}\right\} }}\sum_{N-1\leq k+l+q\leq3\left(N-1\right)}^{N-1}\beta^{k+l+q+1-N}\left\vert V_{n}^{m_{1}m_{2}m_{3}}\right\vert \times\\
\times & e^{c'\left(k+l+q\right)+c\left(k^{2}+l^{2}+q^{2}\right)}e^{-\left(\gamma-\varepsilon'-\eta^{\prime}\right)\left(\left\vert x_{m_{1}}\right\vert +\left\vert x_{m_{2}}\right\vert +\left\vert x_{m_{3}}\right\vert \right)}\nonumber \end{align}
 we find from \eqref{eq:overlap}\begin{align}
 & {\displaystyle \sum\limits _{m_{1}m_{2}m_{3}}}\left\vert V_{n}^{m_{1}m_{2}m_{3}}\right\vert e^{-\left(\gamma-\varepsilon'-\eta^{\prime}\right)\left(\left\vert x_{m_{1}}\right\vert +\left\vert x_{m_{2}}\right\vert +\left\vert x_{m_{3}}\right\vert \right)}\label{eq:bootstrap3}\\
 & \leq V_{\delta}^{\varepsilon,\varepsilon^{\prime}}{\displaystyle \sum\limits _{m_{1}m_{2}m_{3}}}e^{\varepsilon\left(\left\vert x_{n}\right\vert +\left\vert x_{m_{1}}\right\vert +\left\vert x_{m_{2}}\right\vert +\left\vert x_{m_{3}}\right\vert \right)}\nonumber \\
 & \times e^{-\frac{1}{3}\left(\gamma-\varepsilon^{\prime}\right)\left(\left\vert x_{n}-x_{m_{1}}\right\vert +\left\vert x_{n}-x_{m_{2}}\right\vert +\left\vert x_{n}-x_{m_{3}}\right\vert \right)}e^{-\left(\gamma-\eta^{\prime}-\varepsilon'\right)\left(\left\vert x_{m_{1}}\right\vert +\left\vert x_{m_{2}}\right\vert +\left\vert x_{m_{3}}\right\vert \right)}\nonumber \\
 & \leq V_{\delta}^{\varepsilon,\varepsilon^{\prime}}e^{-\left(\gamma-\varepsilon-\varepsilon^{\prime}\right)\left\vert x_{n}\right\vert }{\displaystyle \sum\limits _{m_{1}m_{2}m_{3}}}e^{-\frac{2}{3}\left(\gamma-\eta^{\prime}-\varepsilon\right)\left(\left\vert x_{m_{1}}\right\vert +\left\vert x_{m_{2}}\right\vert +\left\vert x_{m_{3}}\right\vert \right)}=C_{\delta}^{\gamma,\varepsilon,\varepsilon',\eta^{\prime}}e^{-\left(\gamma-\varepsilon-\varepsilon^{\prime}\right)\left\vert x_{n}\right\vert }\nonumber \end{align}
 for $\frac{\gamma}{3}<\left(\gamma-\eta^{\prime}\right).$ Therefore
for $\eta^{\prime}$ sufficiently small, substituting back we get\begin{align}
S_{1} & \leq C_{\delta}^{\gamma,\varepsilon,\varepsilon^{\prime\prime},\eta^{\prime}}e^{-\left(\gamma-\varepsilon-\varepsilon^{\prime}\right)\left\vert x_{n}\right\vert }\sum_{N-1\leq k+l+q\leq3\left(N-1\right)}^{N-1}\beta^{k+l+q+1-N}e^{c\left(k^{2}+l^{2}+q^{2}\right)}e^{c'\left(k+l+q\right)}\label{eq:S_1_bound}\\
 & \equiv C_{\delta}\left(N\right)e^{6cN^{2}}e^{-\left(\gamma-\varepsilon-\varepsilon^{\prime}\right)\left\vert x_{n}\right\vert },\nonumber \end{align}
where we used $c'=cN$, and $C_{\delta}\left(N\right)=O\left(N^{3}\right)$.
For the bound \eqref{eq:S_1_bound} to be violated at least one of
the $c_{m}^{\left(k\right)}$ has to satisfy the inequality \eqref{eq:theorem1}
and the probability for this is bounded by $e^{-c'}e^{-\eta'\left|x_{m}\right|/k}$.
Therefore the probability that \eqref{eq:S_1_bound} will be violated
is bounded by\begin{equation}
e^{-c^{\prime}}{\displaystyle \sum\limits _{k=1}^{N-1}}{\displaystyle \sum\limits _{n=-\infty}^{\infty}}e^{-\eta^{\prime}\left\vert x_{n}\right\vert /k}=e^{-c^{\prime}}{\displaystyle \sum\limits _{k=1}^{N-1}}\left(\frac{2}{1-e^{-\eta^{\prime}/k}}-1\right){\displaystyle \sim\sum\limits _{k=1}^{N-1}}\frac{k}{\eta^{\prime}}\sim\frac{N^{2}}{\eta'}e^{-cN},\label{eq:bootstrap:prob}\end{equation}
where we have expanded the exponent $e^{-\eta^{\prime}/k}$ using
the fact that for large $N$ the sum is dominated by terms with $k\gg1$
and $\eta^{\prime}/k\ll1$. Setting $c^{\prime}=N$ provides the convergence
of the probability with the expansion order, i.e.,\begin{equation}
\Pr\left(\begin{array}{c}
\sum_{m_{1}m_{2}m_{3}}\sum_{N-1\leq k+l+q\leq3\left(N-1\right)}^{N-1}\beta^{k+l+q+1-N}\left\vert V_{n}^{m_{1}m_{2}m_{3}}\right\vert \left\vert c_{m_{1}}^{\ast\left(k\right)}\right\vert \left\vert c_{m_{2}}^{\left(l\right)}\right\vert \left\vert c_{m_{3}}^{\left(q\right)}\right\vert \\
\geq C_{\delta}\left(N\right)e^{6cN^{2}}e^{-\left(\gamma-\varepsilon-\varepsilon^{\prime}-\eta'\right)\left\vert x_{n}\right\vert }\end{array}\right)\leq\frac{const}{\eta'}N^{2}e^{-cN}.\label{eq:bootstrap:term0:bound}\end{equation}
Now we turn to find the point $t_{0}$ defined in \eqref{eq:t0_definition}.
To bound other expressions in \eqref{eq:bootstrap1}, we use\begin{equation}
\left\vert Q_{m}\left(t_{0}\right)\right\vert =2t_{0}\cdot S_{1}\leq M\left(t_{0}\right)e^{-\nu\left\vert x_{m}\right\vert },\label{eq:bootstrap:assumption}\end{equation}
where following \eqref{eq:bootstrap:term0:bound}, \begin{equation}
M\left(t_{0}\right):=2t_{0}C_{\delta}\left(N\right)e^{6cN^{2}}\label{eq:bootstrap:M_definition}\end{equation}
with $\nu=\gamma-\varepsilon-\varepsilon^{\prime}$. In what follows,
unless stated differently, $M$ will mean $M\left(t_{0}\right)$.

First we will bound the linear term in the $Q_{n}$ in \eqref{eq:bootstrap1}.
The sum over $m_{2}$ and $m_{3}$ is bounded similarly to the sums
in the inhomogeneous term and the sum over $m_{1}$ is bounded using
the bootstrap assumption \eqref{eq:bootstrap:assumption}, resulting
in \begin{align}
S_{2}= & {\displaystyle \sum\limits _{m_{2}m_{3}}}\sum_{0\leq k+l\leq2\left(N-1\right)}^{N-1}\beta^{l+k+1}\sum_{m_{1}}\left\vert V_{n}^{m_{1}m_{2}m_{3}}\right\vert \left\vert Q_{m_{1}}^{\ast}\right\vert \left\vert c_{m_{2}}^{\left(l\right)}\right\vert \left\vert c_{m_{3}}^{\left(k\right)}\right\vert \\
 & \leq M{\displaystyle \sum\limits _{m_{2}m_{3}}}\sum_{0\leq k+l\leq2\left(N-1\right)}^{N-1}\beta^{l+k+1}e^{c\left(k^{2}+l^{2}\right)}e^{c'\left(k+l\right)}\sum_{m_{1}}\left\vert V_{n}^{m_{1}m_{2}m_{3}}\right\vert e^{-\nu\left\vert x_{m_{1}}\right\vert }e^{-\left(\gamma-\eta^{\prime}-\varepsilon'\right)\left(\left\vert x_{m_{2}}\right\vert +\left\vert x_{m_{3}}\right\vert \right)}\nonumber \end{align}
 for $\frac{1}{3}\left(\gamma-\varepsilon^{\prime}\right)<\nu$ and
$\frac{\gamma}{3}<\gamma-\eta'$. This is similar to the sum in equation
(\ref{eq:bootstrap3}) and gives a result with a similar dependence
on $\left\vert x_{n}\right\vert $,\begin{eqnarray}
{\displaystyle S_{2}} & \leq & \sum\limits _{m_{2}m_{3}}\sum_{0\leq k+l\leq2\left(N-1\right)}^{N-1}\beta^{l+k+1}\sum_{m_{1}}\left\vert V_{n}^{m_{1}m_{2}m_{3}}\right\vert \left\vert Q_{m_{1}}^{\ast}\right\vert \left\vert c_{m_{2}}^{\left(l\right)}\right\vert \left\vert c_{m_{3}}^{\left(k\right)}\right\vert \\
 & \leq & C_{\delta}\left(N\right)e^{2cN^{2}}Me^{-\left(\gamma-\varepsilon-\varepsilon^{\prime}\right)\left\vert x_{n}\right\vert }.\nonumber \end{eqnarray}
 with the same probability as in (\ref{eq:bootstrap:prob}). Therefore
the linear term in $Q_{n}$ is bounded by the probabilistic bound\begin{equation}
\Pr\left(\begin{array}{c}
\sum_{m_{1}m_{2}m_{3};l,k}\beta^{l+k+1}\left\vert V_{n}^{m_{1}m_{2}m_{3}}\right\vert \left\vert Q_{m_{1}}^{\ast}\right\vert \left\vert c_{m_{2}}^{\left(l\right)}\right\vert \left\vert c_{m_{3}}^{\left(k\right)}\right\vert \\
\geq\beta C_{\delta}\left(N\right)e^{4cN^{2}}Me^{-\left(\gamma-\varepsilon-\varepsilon^{\prime}-\eta'\right)\left\vert x_{n}\right\vert }\end{array}\right)\leq\frac{const}{\eta'}N^{2}e^{-cN},\label{eq:bootstrap:term1:bound}\end{equation}
with $C_{\delta}=O\left(N^{2}\right)$. A similar bound is found for
the third term on RHS of \eqref{eq:bootstrap1}.

The fourth sum of equation \eqref{eq:bootstrap1}\begin{equation}
S_{4}=\sum_{m_{1}m_{2}m_{3};k}\beta^{N+k+1}\left\vert V_{n}^{m_{1}m_{2}m_{3}}\right\vert \left\vert Q_{m_{1}}^{\ast}\right\vert \left\vert Q_{m_{2}}\right\vert \left\vert c_{m_{3}}^{\left(k\right)}\right\vert \end{equation}
is bounded by\begin{equation}
\Pr\left(\begin{array}{c}
{\displaystyle \sum\limits _{m_{1}m_{2}m_{3};k}}\beta^{N+k+1}\left\vert V_{n}^{m_{1}m_{2}m_{3}}\right\vert \left\vert Q_{m_{1}}^{\ast}\right\vert \left\vert Q_{m_{2}}\right\vert \left\vert c_{m_{3}}^{\left(k\right)}\right\vert \\
\geq C_{\delta}\left(N\right)e^{2cN^{2}}\beta^{N+1}M^{2}e^{-\left(\gamma-\varepsilon-\varepsilon^{\prime}-\eta'\right)\left\vert x_{n}\right\vert }\end{array}\right)\leq\frac{const}{\eta'}N^{2}e^{-cN},\label{eq:bootstrap:term2:bound}\end{equation}
with $C_{\delta}=O\left(N\right)$. The last term in \eqref{eq:bootstrap1}
is\begin{equation}
S_{5}=\beta^{2N+1}\sum_{m_{1}m_{2}m_{3}}\left\vert V_{n}^{m_{1}m_{2}m_{3}}\right\vert \left\vert Q_{m_{1}}^{\ast}\right\vert \left\vert Q_{m_{2}}\right\vert \left\vert Q_{m_{3}}\right\vert .\label{eq:bootstrap:term3}\end{equation}
and it is bounded by\begin{equation}
S_{5}\leq\beta^{2N+1}M^{3}e^{-\left(\gamma-\varepsilon-\varepsilon^{\prime}\right)\left\vert x_{n}\right\vert }.\label{eq:bootstrap:term3:bound}\end{equation}
To summarize, substitution of the equality \eqref{eq:bootstrap:assumption}
in \eqref{eq:bootstrap1} to find a time $t_{0}$ for which assumption
\eqref{eq:bootstrap1-2} is valid results in the following inequality
which is satisfied with the probability that is the sum of the probabilities
given by the RHS of \eqref{eq:bootstrap:term0:bound}, \eqref{eq:bootstrap:term1:bound},
\eqref{eq:bootstrap:term2:bound},\begin{equation}
2t_{0}S_{1}\leq t_{0}\cdot S_{1}+t_{0}C_{\delta}\left(N\right)\left(e^{6cN^{2}}+2\beta e^{4cN^{2}}M+4\beta^{N+1}e^{2cN^{2}}M^{2}+8\beta^{2N+1}M^{3}\right)e^{-\left(\gamma-\varepsilon-\varepsilon^{\prime}\right)\left\vert x_{n}\right\vert },\label{eq:bootstrap:final}\end{equation}
with $C_{\delta}=O\left(N^{3}\right)$. Multiplying by $e^{\left(\gamma-\varepsilon-\varepsilon^{\prime}\right)\left\vert x_{n}\right\vert }$
both sides of the inequality \eqref{eq:bootstrap:final} and taking
the infimum with respect to $n$, gives \begin{equation}
\inf_{n}\left(S_{1}e^{\left(\gamma-\varepsilon-\varepsilon^{\prime}\right)\left\vert x_{n}\right\vert }\right)\leq2C_{\delta}\left(N\right)\left(\beta e^{4cN^{2}}M+2\beta^{N+1}e^{2cN^{2}}M^{2}+4\beta^{2N+1}M^{3}\right).\end{equation}
Setting $A_{N}:=e^{2cN^{2}}$, $\tilde{C}:=\inf_{n}\left(S_{1}e^{\left(\gamma-\varepsilon-\varepsilon^{\prime}\right)\left\vert x_{n}\right\vert }\right)$
and using the definition of $M$ (see \eqref{eq:bootstrap:M_definition})
\begin{equation}
\tilde{C}\leq2\beta t_{0}\cdot A_{N}^{5}C_{\delta}+4\beta^{N+1}t_{0}^{2}A_{N}^{7}C_{\delta}^{2}+8\beta^{2N+1}t_{0}^{3}A_{N}^{9}.\label{eq:bootstrap:M}\end{equation}
For sufficiently small $\beta t_{0}$ the first term on the RHS of
\eqref{eq:bootstrap:M} is dominant and therefore,\begin{equation}
t_{0}\geq\frac{\tilde{C}}{2\beta A_{N}^{5}C_{\delta}}.\label{eq:bootstrap_t0}\end{equation}
Note, that $\tilde{C}>0$, since it is an infimum of a sum of positive
quantities, however we do not calculate it explicitly in this paper,
nevertheless it is likely to be of the order of $C_{\delta}A_{N}^{3}$.
This proves that,\[
\left|Q_{n}\left(t\right)\right|\leq M\left(t\right)\cdot e^{-\left(\gamma-\varepsilon-\varepsilon^{\prime}\right)\left\vert x_{n}\right\vert }=2t\cdot C_{\delta}e^{6cN^{2}}e^{-\left(\gamma-\varepsilon-\varepsilon^{\prime}\right)\left\vert x_{n}\right\vert }\]
for times $t\leq t_{0}$ , where $M\left(t\right)$ is given by extending
the definition \eqref{eq:bootstrap:M_definition} by replacing $t_{0}$
by $t$.
\begin{thm}
For $t=O\left(\beta^{-1}\right)$ with $c\cong2$, and assuming that
Conjecture \ref{con:joint_dist} holds \begin{align}
\Pr\left(\left\vert Q_{n}\right\vert \geq2t\cdot C_{\delta}e^{6cN^{2}}e^{-\left(\gamma-\varepsilon-\varepsilon^{\prime}\right)\left\vert x_{n}\right\vert }\right) & \leq\frac{const}{\eta'}N^{2}e^{-cN}\label{eq:final_theorem}\end{align}

\end{thm}
The bound on the probability of statement (\ref{eq:final_theorem})
is calculated by summing the RHS of (\eqref{eq:bootstrap:term0:bound},\eqref{eq:bootstrap:term1:bound},\eqref{eq:bootstrap:term2:bound}).
The contribution of the remainder term is \begin{equation}
\left|\beta^{N}Q_{n}\right|\leq const\cdot e^{6cN^{2}+N\ln\beta+\ln t}e^{-\left(\gamma-\varepsilon-\varepsilon^{\prime}\right)\left\vert x_{n}\right\vert }.\end{equation}
Note that for a given $t$ and $\beta$ there is an optimum $N$ for
which the remainder is minimal. Additionally, for any fixed time and
order $N$, $\lim_{\beta\rightarrow0}\left|\beta^{N}Q_{n}\right|/\beta^{N-1}=0$,
which shows that the series is in fact an asymptotic one \cite{Erdelyi1956}.

\section{\label{sec:Summary}Summary}

In this paper a perturbation expansion in powers of $\beta$ was developed
(Sections \ref{sec:Organization-of-the}, \ref{sec:Elimination-of-secular})
for the solution of the NLSE with a random potential. It required
the removal of the secular terms for this problem. To best of our
knowledge it is the first time it was done for a multivariate problem.
The quality of the expansion to the second order was tested in Section
\ref{sec:Numerical-results}. In Section \ref{sec:The-entropy-problem}
it was shown that the number of terms grows exponentially with the
order. In Section \ref{sec:Bounding-the-general} a probabilistic
bound on the general term \eqref{eq:theorem1} was derived. It relies
on the Conjecture \ref{con:joint_dist}. The resulting bound was tested
numerically. Finally, a bound on the remainder was obtained for a
finite time, showing that the series is asymptotic. For time shorter
than $t_{0}$ which is given by \eqref{eq:bootstrap_t0} there is
a front $\bar{x}\left(t\right)\propto\ln t$ such that for $x_{n}>\bar{x}\left(t\right)$
both the remainder, $\beta^{N}Q_{n}\left(t\right)$ and $c_{n}\left(t\right)$
are exponentially small.

The work leaves several open problems that should be subject of further
research:
\begin{enumerate}
\item Turning the perturbation theory developed in the present work into
a practical method for solution of the NLSE and similar nonlinear
differential equations. The control on the error should be obtained
using the methods presented in Sec. \ref{sec:Bounding-the-remainder}.
\item Can the front $\bar{x}\left(t\right)\propto\ln t$ be found for arbitrarily
long times ?
\item The asymptotic nature of the series. Is it just an asymptotic series
or a convergent one ?
\item If the series is asymptotic can it be resummed ?
\item The $e^{ck^{2}+c'k}$ in \eqref{eq:theorem1} results from the repeated
use of the Hölder inequality and a very generous bound on the probability
distributions. An effort should be made to improve it.
\item There are various properties of the Anderson model that have been
used here. Some of them were tested numerically. It would be of great
value if Conjecture \ref{con:joint_dist}, Corollary \ref{cor:main_theorem}
and Conjecture \ref{con:central_limit} were rigorously established,
even at the limit of a strong disorder. In the present work we rely
only on Corollary \ref{cor:main_theorem} (that was tested numerically,
see Fig. \ref{fig:mean_of_f}). The rigorous proof of the unimodality
of $\gamma\left(E\right)$ for the uniform distribution of the random
potentials, $\varepsilon_{x}$, may also be useful.
\item It would be very useful if Conjecture \ref{con:e_prime} could be
rigorously obtained.
\end{enumerate}
We enjoyed many extensive illuminating and extremely critical discussions
with Michael Aizenman. We also had informative discussions with S.
Aubry, V. Chulaevski, S. Flach, I. Goldshield, M. Goldstein, I. Guarneri,
M. Sieber, W.-M. Wang and S. Warzel. This work was partly supported
by the Israel Science Foundation (ISF), by the US-Israel Binational
Science Foundation (BSF), by the USA National Science Foundation (NSF),
by the Minerva Center of Nonlinear Physics of Complex Systems, by
the Shlomo Kaplansky academic chair, by the Fund for promotion of
research at the Technion and by the E. and J. Bishop research fund.
The work was done partially while SF was visiting the Institute of
Mathematical Sciences, National University of Singapore in 2006 and
the Center of Nonlinear Physics of Complex Systems in Dresden in 2007
and while YK visited the department of Mathematics at Rutgers University
and while AS was visiting the Lewiner Institute of Theoretical Physics
at the Technion in 2007. The visits were supported in part by these
Institutes.

\section*{Appendix}

The average Lyapunov exponent was calculated numerically using the
transfer matrix technique for a uniform distribution defined in \eqref{eq:uniform_dist}.
The results are presented in Fig. \ref{fig:gamma_E}. It can be seen
that $\gamma\left(E\right)$ is unimodal. %
\begin{figure}
\includegraphics[width=0.95\textwidth]{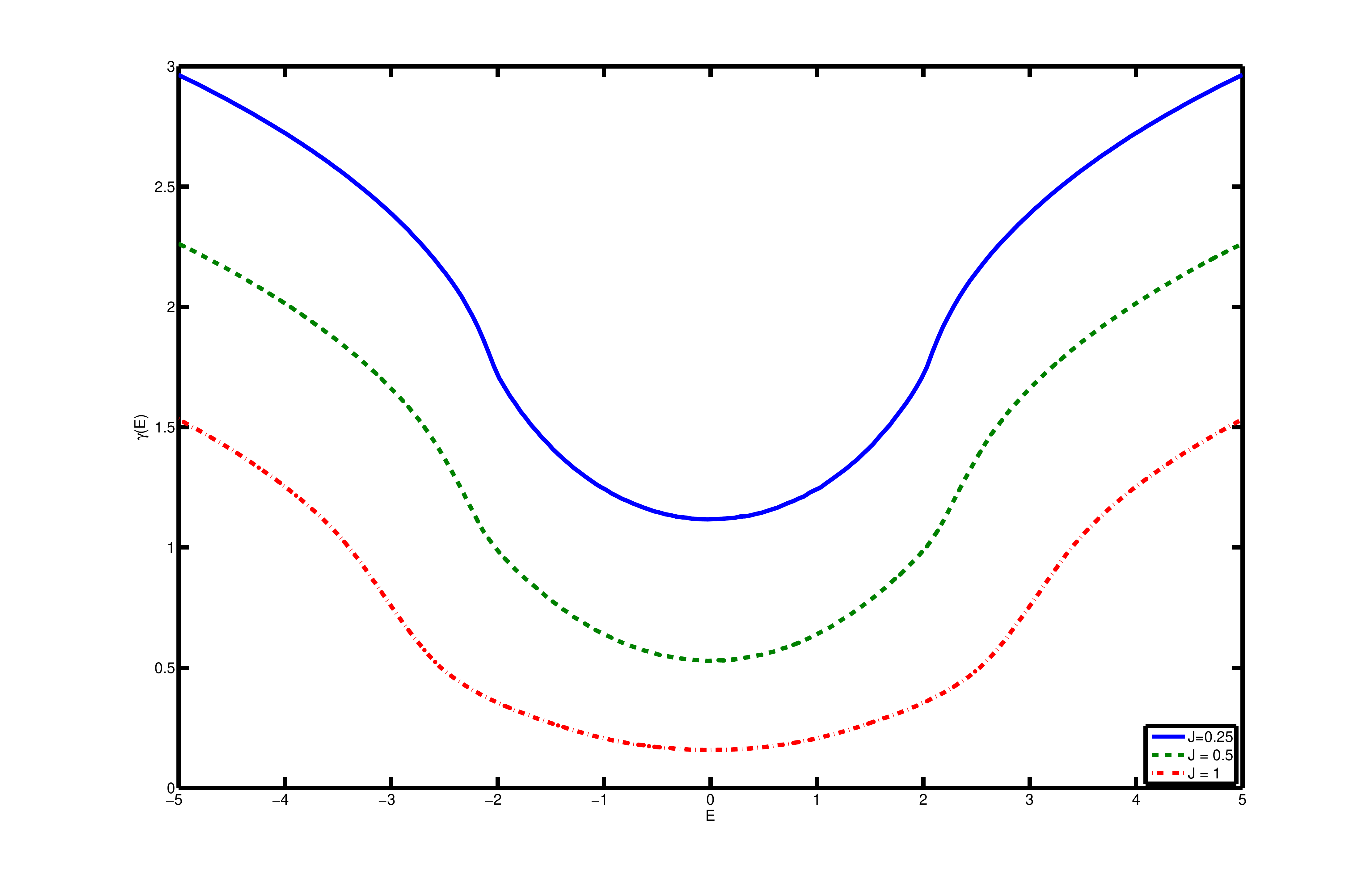}

\caption{\label{fig:gamma_E}The Lyapunov exponent as a function of energy
for the uniform distribution defined in \eqref{eq:uniform_dist}.
The solid (blue) line is for $J=0.25$, the dashed (green) line is
for $J=0.5$ and the dot-dashed line (red) is for $J=1$.}

\end{figure}

\end{document}